\documentclass{amsart}
\usepackage{amsmath}
\usepackage{graphicx}
\usepackage{amsfonts}
\usepackage{amssymb}

\newtheorem{theorem}{Theorem}

\newtheorem{definition}[theorem]{Definition}

\newtheorem{lemma}[theorem]{Lemma}

\newtheorem{proposition}[theorem]{Proposition}
\newtheorem{remark}[theorem]{Remark}

\newenvironment{oss}{\begin{remark} \begin{rm}}{\end{rm} \end{remark}}

\newcommand{\ts}{{\mathcal T}}

\newcommand{\sss}{{\mathcal S}}

\newcommand{\F}{{\mathcal F}}
\newcommand{\GL}{\mathcal{GL}}
\newcommand{\XY}{\mathcal{XY}}
\newcommand{\SD}{\mathcal{SD}}
\newcommand{\G}{\mathcal{G}}
\newcommand{\W}{W}
\newcommand{\bGL}{\mathbb{GL}}
\newcommand{\bXY}{\mathbb{XY}}
\newcommand{\bSD}{\mathbb{SD}}

\newcommand{\R}{{\mathbb R}}

\newcommand{\Z}{{\mathbb Z}}

\newcommand{\C}{{\mathbb C}}

\newcommand{\Om}{\Omega}

\newcommand{\asexy}{{\mathcal{AXY}_\e}}

\newcommand{\asesd}{{\mathcal{ASD}_\e}}

\newcommand{\Torus}{\mathbb{T}}

\newcommand{\oez}{\Om_\e^0}

\newcommand{\oeu}{\Om_\e^1}
\newcommand{\oed}{\Om_\e^2}

\newcommand{\weak}{\rightharpoonup}

\newcommand{\e}{\varepsilon}
\newcommand{\en}{{\varepsilon_n}}
\newcommand{\de}{{\delta_{\varepsilon}}}

\newcommand{\f}{\varphi}
\newcommand{\dd}{{\bf{d}}}

\newcommand{\pr}{P}

\newcommand{\curl}{\text{curl }}

\newcommand{\dist}{\text{dist}}

\newcommand{\res}{\mathop{\hbox{\vrule height 7pt width .5pt depth 0pt
\vrule height .5pt width 6pt depth 0pt}}\nolimits}

\title
[Complex Ginzburg Landau energies, XY spin systems and 
dislocations] {Variational equivalence between Ginzburg-Landau, XY
spin systems and screw dislocations energies}
\author[R. Alicandro]
{Roberto Alicandro}
\author[M. Cicalese]
{Marco Cicalese}
\author[M. Ponsiglione]
{Marcello Ponsiglione}
\address[R. Alicandro]{DAEIMI, Universit\`a di Cassino via Di Biasio 43, 03043 Cassino (FR), Italy}
\email[R. Alicandro]{alicandr@unicas.it}
\address[M. Cicalese]{Dipartimento di Matematica e Applicazioni ``R.~Caccioppoli",
Universit\`a di Napoli ``Federico II" via Cintia, 80126 Napoli,
Italy} \email[M. Cicalese]{cicalese@unina.it}
\address[Marcello Ponsiglione]{Dipartimento di Matematica ``G. Castelnuovo", Universit\'a di Roma ``La Sapienza",
Piazzale A. Moro 2, 00185 Roma, Italy} \email[M.
Ponsiglione]{ponsigli@mat.uniroma1.it}

\begin{document}
\vskip .2truecm

\begin{abstract}
\small{We introduce and discuss discrete two-dimensional models
for $XY$ spin systems and screw dislocations in crystals. We prove
that, as the lattice spacing $\e$ tends to zero, the relevant
energies in these models behave like a free energy in the complex
Ginzburg-Landau theory of superconductivity, justifying in a
rigorous mathematical language  the analogies between screw
dislocations in crystals and vortices in superconductors.
To this purpose, we introduce a notion of asymptotic variational
equivalence between families of functionals in the framework
of $\Gamma$-convergence. We then prove that, in several scaling regimes, the complex
Ginzburg-Landau, the $XY$ spin system and the screw dislocation
energy functionals are variationally equivalent. Exploiting such an equivalence
between dislocations and vortices, we can show new
results concerning the asymptotic behavior of screw dislocations
in the $|\log\e|^2$ energetic regime.

}
\end{abstract}
\maketitle

\vskip -.5truecm
{\small \tableofcontents}
\section{Introduction}
Since the pioneering papers of Berezinskii \cite{berezinskii},
Kosterlitz \cite{K} and Kosterlitz and Thouless \cite{KT}, there
has been a great effort in studying physical systems exploiting
BKT-phase transitions; {\it i.e.}, phase transitions mediated by
the formation of topological singularities of the order parameter.
This type of phase transitions characterizes several physical
phenomena such as superfluidity, superconductivity and plasticity
(see \cite{Kle_Lav}, \cite{Lo1}, \cite{Lo2}, \cite{Mermin}, \cite{Simons}), while vortices in
superconductivity and $XY$ spin systems, as well as screw
dislocations in crystals, provide three paradigmatic examples of
singularities.
\par
The phenomenological analogies shown by these apparently far physical systems have been
pointed out many times in the physical community. In the language of statistical mechanics it has also been rigorously proven that
these systems belong to the same universality class all of them sharing a BKT-type phase transition. Roughly speaking it is known
that above some temperature threshold, these systems undergo a
phase transition to a disordered state in which the topological singularities are unbound. Under that threshold
the correlation length exponentially decays and the singularities bind together
and interact through complex and mostly unknown pheno\-mena involving many interacting scales. Such
a complex behavior is the main reason why a detailed analysis of the ground states of these systems turns out to be a
non trivial task. Moreover, the above description explains, to some extent, why a qualitative approach
to the study of the thermodynamic limit of these systems, such as the celebrated Ginzburg-Landau theory, has been so successfully exploited.
\par
In this paper we are concerned with the problem of describing some
relevant properties of the ground states of these systems in the
thermodynamic limit. We aim to provide a unifying mathematical point of view, based on a variational equivalence argument, to study the asymptotic behavior of the ground states of different models that share the same geometrical
and topological qualitative features.
More specifically our purpose is two-fold.
On one hand,  we want to reinterpret  several known results about the asymptotic behavior of the
ground states of such models, proving that the corresponding free energies are indeed equivalent from a variational point of view. On the other hand, taking advantage of this equivalence, we want
to exploit some of the results currently proved only through a
phenomenological Ginzburg-Landau analysis, to obtain new
results in different contexts.
\par
As in \cite{KT}, among the physical systems exploiting topological
type phase transitions we focus on two-dimensional systems and we
choose two paradigmatic examples: screw dislocations ($SD$) and $XY$
spin systems. We will introduce two basic discrete models, both
constructed on $\e\Z^2\cap\Omega$, where $\Omega\subset\R^2$ is a  bounded open
set. Their order parameters are a unit
vectorial spin field for the $XY$ model:
$i\in\e\Z^2\cap\Omega\mapsto v(i)\in\R^2$ such that $|v(i)|=1$, and a scalar
displacement field for the $SD$ model:
$i\in\e\Z^2\cap\Omega\mapsto u(i)\in\R$. For a given configuration
of spins or displacements, the energies of these systems are given
by
\begin{equation*}
XY_\e(v):= \frac{1}{2} \sum_{i,j\in\e\Z^2\cap \, \Omega:\
|i-j|=\e}
 |{v(i)-v(j)}|^2
\end{equation*}
and
\begin{equation*}
SD_\e(u):= \frac{1}{2} \sum_{i,j\in \e\Z^2\cap \, \Omega:\
|i-j|=\e} \text{dist}^2(u(i)-u(j),\Z).
\end{equation*}
With the energies written in this form, the coarse-graining
analysis now amounts to study the limit, as $\e\to 0$, of (some
scaled version of) $XY_\e$ and $SD_\e$. To this purpose, in the
physical literature, it is customary to perform a so-called
Ginzburg-Landau (GL) analysis (see \cite{Kle_Lav} and \cite{Simons} for
an introduction to the subject and some applications). The main
ansatz of this approach (based on heuristic scaling and symmetry
type arguments) is to assume that some of the interesting features
of the thermodynamic limit of the original functionals can be
obtained by studying the limit, as $\e\to 0$, of a family of
so-called complex Ginzburg-Landau energies. These energies have as
order parameter a vectorial field $x\in\Omega\mapsto w(x)\in\R^2$
and are defined as
\begin{equation}\label{intro_GL}
GL_\e(w):= \int_{\Om} \frac{1}{2} |\nabla w|^2 +
\frac{1}{\e^2} (1-|w|^2)^2.
\end{equation}
This kind of functionals has been originally introduced as a
phenomenological phase-field type free-energy of a superconductor,
near the superconducting transition, in absence of an external
magnetic field. Here  the order parameter $w$ describes how deep
the system is into the superconducting phase and the scale $\e$ is
proportional to the coherence length of the superconductor.
\par
The Ginzburg-Landau functionals have deserved a great attention by
the physical community. The first rigorous mathematical approach
to the limit as $\e\to 0$ of the solutions to the Euler-Lagrange
equation of \eqref{intro_GL} has been made by Bethuel, Brezis and
H\'elein in \cite{BBH}. Since their paper a great effort has been
done to study the asymptotic behaviour of minimizers of the
Ginzburg-Landau energy both from the PDEs and the Calculus of
Variations points of view. In particular in dimension two Jerrard
and Soner in \cite{JS} (see also Alberti, Baldo and Orlandi \cite{ABO} for the 
generalization to any dimension) have proved a $\Gamma$-convergence result for
$\frac{GL_\e(w)}{|\log\e|}$. In their analysis  the relevant tool to track energy concentration 
is the asymptotic behavior of the Jacobians $J(w_\e)$ of sequences $(w_\e)$  equi-bounded in   energy.
In particular they prove that, up to subsequences, $J(w_\e)$ converges to a finite sum of Dirac deltas
whose support represents the vortex-like singularities of the limit field and that the $\Gamma$-limit is proportional to the
number of such singularities.

Only recently  a similar analysis in the context of spin systems and dislocations has  attracted much
attention in the  mathematical community and it has been carried on both in a continuous framework (see \cite{CG}, \cite{CL}, \cite{GLP}, \cite{GM2}) and in a discrete setting (see \cite{AlBrCi}, \cite{AC},  \cite{Alcigl}, \cite{Po}). As a further remark we underline that the asymptotic analysis of spin systems and discrete dislocations energy functionals is itself part of a wider interest in the discrete-to-continuum limits for more general models (see for example
\cite{BlLBLiARMA}, \cite{BlLBLirev} and \cite{brhandbook} (Chapter $11$) for
a review on this subject). \par In  \cite{AC} and
in \cite{Po} a $\Gamma$-convergence result for $XY$-spin systems
and for the $SD$ model is given in the $|\log\e|$ scaling regime. Here the $\Gamma$-convergence analysis is performed with respect to the convergence of the Jacobians of a suitable affine interpolation of the spin
variable $v$ for the $XY$ model, and with respect to the convergence of a suitable discrete notion of the ${\it curl}$ of the strain field $u$ for the $SD$ model. Roughly speaking, gathering together the main results of these two
papers, the following relations hold:
$$
\Gamma\hbox{-}\lim_{\e\to 0} \frac{GL_\e(w)}{|\log\e|} =
\Gamma\hbox{-}\lim_{\e\to 0} \frac{XY_\e(v)}{|\log\e|}=
\Gamma\hbox{-}\lim_{\e\to 0} \frac{4\pi^2 SD_\e(u)}{|\log\e|}.
$$
Motivated by this chain of equalities, we were led to ask whether one could provide, in the framework of $\Gamma$-convergence, a unifying mathematical 
point of view to rigorously relate the asymptotic behavior of these models. 
Our purpose is to prove that the asymptotic equivalence of  these models, in terms of $\Gamma$-convergence,  can be push  forward to any $|\log\e|^h$ scaling regime with $h\geq 1$. More precisely,  we show that 
\begin{equation}\label{intro:result}
\Gamma\hbox{-}\lim_{\e\to 0} \frac{GL_\e(w)}{|\log\e|^h}=
\Gamma\hbox{-}\lim_{\e\to 0} \frac{XY_\e(v)}{|\log\e|^h}=
\Gamma\hbox{-}\lim_{\e\to 0} \frac{4\pi^2 SD_\e(u)}{|\log\e|^h}.
\end{equation}
In this way we rigorously obtain the equivalence of
some mean field models for vortices, spin systems and dislocations, according with experimental evidence (see \cite{R} for a recent overview of the analogies  between the mean fields in these models). In particular, we obtain a rigorous justification to the $GL$  analysis of the thermodynamic limits of the $XY$ and $SD$ models in the $|\log\e|^h$ energetic regimes.
\par
To prove \eqref{intro:result}, we look for a relation between the order parameters of the different 
models, at proper mesoscopic scales, which allow us to compare the three families of energy
functionals  in the $|\log\e|^h$ scaling regime for every $h\geq 1$. 
This has led us to introduce a notion of {\it variational equivalence} between families of functionals 
(see Definition \ref{equivalence}). To explain the meaning
of such a notion, let us suppose that we are given
two families of energies $(F_\e)$ and $(G_\e)$ depending on
a small parameter $\e$ standing for an interaction scale. Then we say that $(F_\e)$ and $(G_\e)$ are variationally
equivalent if there holds that $G_\e\preceq F_\e$ and $F_\e\preceq G_\e$. By  
$G_\e\preceq F_\e$ we mean that, for any given family of order parameters $(p_\e)$ such that
$F_\e(p_\e)\leq C$, there exists another scale $\delta_\e$ 
and a family of order parameters $q_\e$ such that
$q_\e$ is closer and closer to  $p_\e$ and $G_{\delta_\e}(q_\e)\leq F_\e(p_\e)+O(1)$. 
Roughly speaking, we are saying that
the two variational models whose energies are given by $F_\e$ and $G_\e$ describe the same phenomena if looked at
proper interaction scales and by suitably choosing the order
parameters on these scales. The main feature of this notion, is
that variational equivalent families of functionals share the same
$\Gamma$-limit and share the so called {\it equi-coercivity}
property (see Theorems \ref{ordiga} and \ref{consequi}).
\par
We remark that our notion departs from the concept
of asymptotically equivalent functionals (at a certain order)
introduced by Braides and Truskinovsky in \cite{BrTr}. Roughly speaking, 
according to their definition, given $\alpha\geq 0$
the functionals $F_\e$ and $G_\e$ are said to be equivalent at order $\alpha$ if 
$\frac{F_\e}{\e^\alpha}$ and $\frac{G_\e}{\e^\alpha}$ have the same $\Gamma$-limit. In particular our 
definition turns out to agree with the latter at order zero provided $(F_\e)$ and $(G_\e)$ are equi-coercive. On
the other hand, to simplify
matter, the purpose of the authors in \cite{BrTr} is to introduce a formalism to build up
variational models that share the same $\Gamma$-limit up to a
certain given scaling order, and then to study the properties that such a convergence enjoys 
with respect to a given family of parameters specific
of the considered theory. Our aim is instead to deduce 
$\Gamma$-convergence and compactness results for the family $G_\e$ from 
the same results for the equivalent family $F_\e$.
\par
In Theorem \ref{mainthm} we prove the equivalence between the families of functio\-nals
$(\frac{GL_\e}{|\log\e|^h})$, $(\frac{XY_\e}{|\log\e|^h})$ and $(\frac{4\pi^2
SD_\e}{|\log\e|^h})$ (see also Theorem \ref{mthmj} for $h=2$). The way our variational equivalence is proved
provides an interesting identification between the order
parameters and the corresponding singularities of the different models. For instance, the
identification underlying the equivalence between $\frac{GL_\e}{|\log\e|^h}$ and
$\frac{XY_\e}{|\log\e|^h}$  relies on suitable interpolation procedures which allow us to
pass from the discrete order parameter of the $XY$ model to the
continuous one of the $GL$ model. Analogously the displacement field
in the $SD$ model is identified with the phase function of the
$XY$ order parameter (see also Remark \ref{ident}). These
identifications of the order parameters clearly induce an
identification of the corresponding singularities (see picture \ref{fig_intro}).
\begin{figure}\label{fig_intro}
\begin{center}
\includegraphics[scale=.5 ]{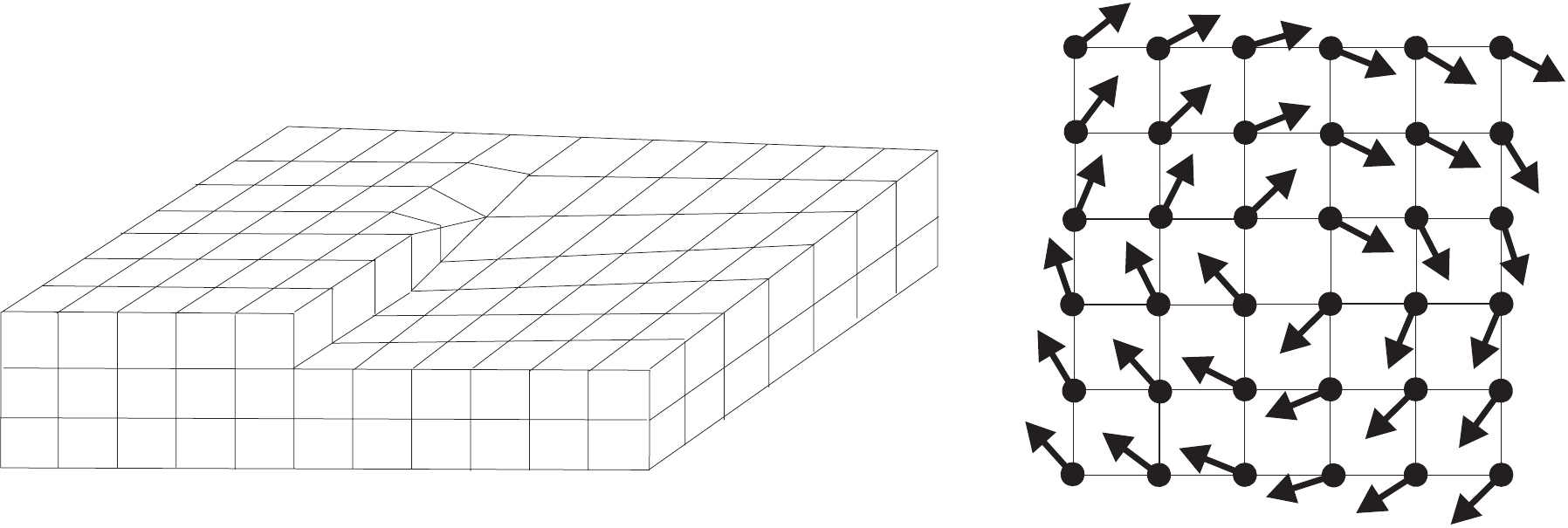}
\caption{A screw dislocation ``corresponds" to a vortex of spins.}
\end{center}
\end{figure}
In the proof of Theorem \ref{mainthm}, we have to make sure that the identification of the order
parameters in the different models produces small perturbations in
the corresponding free energy densities. This is easily checked
far from the singularities, while to control the error near the
singularities we need to introduce suitable reparametrizations
$\de$ of the correlation length, {\it i.e.,} we have to look at
the models at suitable meso-scales $\de$. Finally, during these
identifications, we have also to control the distance, measured in
a suitable topology, of the corresponding singularities. This
analysis involves notions of geometric measure theory, and the
arguments used in the proofs are close to those used to prove density of polyhedral 
boundaries in the space of integer currents in \cite{F} and also exploited in \cite{ABO}. 
\par
Taking advantage of this equivalence principle,  we are able to
export many of the  known results in the theory of $GL$ vortices to
the framework of $SD$ models. Indeed the $\Gamma$-convergence
results in \cite{JS2} in the $|\log\e|^2$ energetic regime,
together with \eqref{intro:result} for $h=2$, leads to new
asymptotic results in the context of dislocations and spin systems
when the number of defects grows logarithmically as $\e$ goes to
zero (see Theorem \ref{mthmsd2}). 
\par
The $|\log\e|^2$ energetic regime  has been already considered in the vectorial
context of homogenizing edge dislocations in \cite{GLP}, within  a
core {\it radius approach}, under the assumption that the
dislocations have a minimal distance of the order of a suitable
meso-scale. The $\Gamma$-convergence analysis done in \cite{GLP}
provides a macroscopic model for plasticity, in agreement with the phenomenological  strain gradient theory for plasticity introduced in \cite{FH}. Moreover, the limit  energy is compatible with the experimental evidence of the concentration of dislocations on
lines, usually referred to as {\it dislocation walls} (we refer  to \cite{CO} for a variational model describing dislocation patterns in crystals). In view of
\eqref{intro:result}, we extend the $\Gamma$-convergence analysis
done in \cite{GLP} to our completely discrete setting, without any
kinematic assumption on the mutual distance of the screw dislocations. In this way, we derive a strain gradient model for plasticity in the scalar setting of anti-planar
elasticity, starting from a completely discrete and basic model of
screw dislocations.
\par
Finally, let us mention that, as a byproduct of our equivalence
result \eqref{intro:result}, taking into account the
$\Gamma$-convergence result proved in \cite{Po}, we obtain (see
Remark \ref{newproof}) a new proof of compactness and
$\Gamma$-convergence for two-dimensional $GL$ functionals,
independent of the proof given in \cite{JS} and in \cite{S}.
\par
Our method can be clearly exploited  in many interesting
directions: first, one can investigate the equivalence of lower
order terms for the mo\-dels we have discussed so far,  comparing
the so called renormalized energy for vortices and dislocations
within a $\Gamma$-convergence analysis, in the spirit of the
theory of development by $\Gamma$-convergence introduced by Braides and Truvskinovsky in \cite{BrTr}.
Moreover, one can consider the case of edge dislocations,
or, more in general, the case of three dimensional models. Indeed,
the results of this paper provide a first step in the effort of
making a link between material dependent models for dislocations
and phenomenological Ginzburg-Landau approaches. We believe  that
our arguments could give efficient hints to build up material
dependent Ginzburg-Landau energies, taking into account kinematic
constraints and elasticity constants specific of the crystal.
Moreover, exploiting our variational equivalence arguments  in the
three dimensional problem (e.g., in a cubic crystal) would bring new light on interesting
mathematical questions regarding compactness properties and
asymptotic behaviour of generalized Ginzburg-Landau functionals,
the target space being a three-dimensional torus, and the
singularities being rectifiable currents with multiplicity in the
group $\Z^3$ (see Section \ref{comments}).
\par
The paper is organized as follows. In Section \ref{overview} we
introduce the discrete models for spin systems and screw
dislocations, while in  Section \ref{tosi} we introduce the
corresponding topological singularities. In Section \ref{vargu} we
describe our variational argument, that will be used in Section
\ref{mainsec} to prove the variational equivalence between
$GL_\e$, $XY_\e$ and $SD_\e$ models. Such an equivalence will be
specialized in Section \ref{New SD} in order to present new
results in the asymptotic analysis of screw dislocations. In Section \ref{comments}, we will comment the
results achieved in this paper suggesting further extensions, concerning,
for instance, the core radius approach to the singularities. Finally we will propose, in the case of three
dimensional elasticity, a material dependent
Ginzburg-Landau type model for dislocations in a cubic lattice. 

\section{Overview of the models}\label{overview}
In this Section we briefly describe the models of Ginzburg-Landau
vortices, of $XY$-spin systems and of screw dislocations. We will
provide a detailed description of the latter in order to define
the physical quantities involved in the model and needed to
correctly describe the new results in the framework of screw
dislocations contained in Section \ref{New SD}.

For the time being $\Om\subset \R^2$ is a bounded open set with Lipschitz
 boundary, representing the domain of definition of the relevant fields in these models.
For the sake of simplicity, we will also assume  that $\Om$ is
star-shaped with respect to the origin. We stress that with some
minor technical effort in our proofs, such assumption can be
removed.

\subsection{Ginzburg Landau functionals}
Let us introduce the family of the so-called complex
Ginzburg-Landau functionals $GL_\e:W^{1,2}(\Om;\R^2)\to [0,+\infty)$,
 defined as
\begin{equation}\label{GL}
GL_\e(w):=\left(  \int_{\Om} \frac{1}{2} |\nabla w|^2 +
\frac{1}{\e^2} \W(w) \right),
\end{equation}
where $W(x):=(1-|x|)^2$. Here $w$ represents the order parameter
of the model, describing how deep the material is in the
superconductive phase, $\e$ is a length-scale parameter, usually
referred to as the {\it coherence length} while
$GL_\e$ is the corresponding free energy of the system.

\begin{remark}{\rm
We make this explicit choice for $W=(1-|x|)^2$ because it
simplifies some computation. The following standard hypotheses
would suffice to perform our analysis: $W\in C(\R^2)$ such that
$W(x)\geq 0$, $W^{-1}\{0\}=S^1$ and
$$
\liminf\limits_{|x|\to 1} \frac{W(x)}{(1-|x|)^2}>0,\quad\liminf\limits_{|x|\to \infty}\frac{W(x)}{ |x|^2}>0.
$$}
\end{remark}

\subsection{The discrete lattice}
Here we introduce the discrete objects and notations we will use
in the sequel.
\par For every positive $\e>0$,
we set $\oez:= \e\Z^2 \cap \Om$, representing the reference
lattice. We will denote by $\oeu:= \{(i,j) \in \oez\times \oez:\,
|i-j|=\e, \, i\le j \}$ the class of nearest neighbors in $\oez$
(where $i\le j$ means  that $i_l\le j_l$ for $l\in\{1,2\}$). The
class of cells contained in $\Om$ is labelled by the set $\oed:=
\{i \in \oez :\, i+ [0,\e]^2 \subset \Om \}$. Finally, we set
$$
\Om_\e := \left\{ \bigcup_{i\in\oed} i+ [0,\e]^2 \right\}.
$$
In the following, we will extend the use of these notations to any
given open subset $A$ of $\R^2$.

\subsection{XY spin systems}
Here we recall the model of $XY$ spin system following the approach
in \cite{AC} (see also \cite{Simons} for a general introduction to the model). 
First we  introduce the class of admissible fields,
\begin{equation}\label{asev}
\asexy:=\{v:\oez \to \sss^1\},
\end{equation}
where $\sss^1$ denotes the set of unit vectors in $\R^2$. The
family of functionals $XY_\e:\asexy\to\R$ are defined by
\begin{equation}\label{GLd}
XY_\e(v):= \frac{1}{2} \sum_{(i, j) \in \oeu}
 |{v(i)-v(j)}|^2.
\end{equation}
We refer the interested reader to \cite{AC} for the derivation of
these energies by proper scaling of the $XY$ energies written in
the usual form 
$$-\sum_{(i,j)\in\oeu}\e^2 \langle v(i),v(j)
\rangle, 
$$
where $\langle a, b \rangle $ denotes the scalar product
between the vectors $a$ and $b$.

\subsection{Screw dislocations}
Here we introduce a basic discrete model for screw
dislocations, inspired by the  approach  introduced in \cite{ArOr}
and revisited in \cite{Po}. The displacement, in this discrete
anti-planar setting, is a function $u:\ \oez \to \R$. We denote
the class of all admissible displacements by
\begin{equation}\label{aseu}
\asesd:=\{u:\oez \to \R\}.
\end{equation}
We focus here on linearized elasticity, and we  consider the model
case of nearest neighbors interactions, so that the discrete
elastic energy corresponding to any displacement $u$, in absence
of dislocations,  is given by (we fix the shear modulus
$\mu=\frac{1}{2}$)
\begin{equation}\label{discela}
E^{el}_\e(u):= \frac{1}{2} \sum_{(i,j)\in\oeu}
 |u(i)-u(j)|^2.
\end{equation}
It is convenient to introduce also the notion of discrete gradient
$\dd u$, defined on the nearest neighbors (namely, the bonds of
the lattice),  by $\dd u_{i,j} = u(j) - u(i)$, for every
$(i,j)\in\oeu $. With respect to the discrete gradient $\dd u$,
the energy \eqref{discela} reads like
\begin{equation}\label{discelad}
E^{el}_\e(u):= \frac{1}{2} \sum_{(i,j)\in\oeu}
 |\dd u_{i,j}|^2.
\end{equation}

To introduce the dislocations in this framework, we adopt the
point of view of the additive decomposition of the gradient of the
displacement in an elastic part, the strain, and a plastic part,
following the formalism of the discrete pre-existing strains as
in \cite{ArOr} and \cite{Po}. More precisely, a pre-existing strain
is a function $\beta^p$ representing the plastic part of the
strain defined on pairs of nearest neighbors and valued in
$\Z |b|$. Here $b$ represents the so called Burgers vector which is
characteristic of the crystal. In principle $|b|$ should be of
order $\e$, but  up to a further re-scaling in the energy
functionals, we can fix from now on $|b|=1$.
\par
Here the idea  is that the plastic strain $\beta^p$ does not store
elastic energy and hence it has to be subtracted to the gradient
of the displacement in order to obtain the so called elastic
strain. In view of this additive decomposition $\nabla u = \beta^p
+ \beta^e$, we have that the elastic strain $\beta^e$ is not
curl-free (in a suitable discrete sense). In Section \ref{tosi}
we will introduce the quantity $\mu:= \curl \beta^e = - \curl
\beta^p$, that measures the degree of incompatibility of the
elastic strain $\beta^e$ from being a discrete gradient, and
represents the discrete screw dislocations in the crystal lattice.
\par
Summarizing, the elastic energy corresponding to the decomposition
$\nabla u = \beta^p + \beta^e$ is given by
$$
E^{el}_\e(u):= \frac{1}{2} \sum_{(i,j)\in\oeu} |\beta^e_{i,j}|^2=
\frac{1}{2} \sum_{(i,j) \in \oeu} |\dd u_{i,j}- \beta^p_{i,j}|^2,
$$
and whenever the dislocation density $\curl \beta^e = - \curl
\beta^p$ is non zero, then the corresponding  elastic energy is
also non zero.
\par
Given a displacement $u$, we can minimize the corresponding
elastic energy with respect to the plastic strain. Since by our
kinematic assumption the plastic strain takes values in $\Z$, it
is clear that the optimal $\beta_u^p$ is obtained projecting ${\bf d}u$
on $\Z$. More precisely, let $\pr:\R\to \Z$ be the projection
operator defined by
\begin{equation}\label{pr}
\pr(t) = \text{argmin} \{|t-s|, s\in \Z\},
\end{equation}
with the convention that, if the argmin is not unique, then we
choose that with minimal modulus. Then we have $\beta_u^p = \pr \dd
u$, in the sense that, for all $(i,j)\in \Omega_\e^1$,
$(\beta_u^p)_{i,j} = \pr (\dd u_{i,j})$.
 In this way, we obtain the elastic energy functionals
$SD_\e:\asesd\to\R$ defined by
\begin{equation}\label{SDd}
SD_\e(u):=\frac{1}{2} \sum_{(i,j) \in \oeu}
 |(\beta_u^e)_{i,j}|^2 =
\frac{1}{2} \sum_{(i,j) \in \oeu}
 \dist^2\big(u(i)-u(j),\Z\big).
\end{equation}

We notice that, defining $\tilde u:= \e u$ as the physical
displacement, we have
$$
{\e^2 SD_\e(u) }=\frac{1}{2} \sum_{(i,j) \in \oeu}
 \text{dist}^2\big(\tilde u(i)- \tilde u(j),\e \Z\big),
 $$
where the r.h.s. now reads as a term penalizing a misfit slip in
the crystal, according with Pierls-Nabarro theories.
\par
Indeed, another way of understanding \eqref{SDd} is the following.
In our anti-planar setting, the crystal can displace only in the
vertical direction, and each vertical line of atoms has to
displace rigidly. The discrete gradient $\dd \tilde u$ measures
the difference of the displacement of two near lines of atoms. If
$\dd \tilde u$ is of order $\e$, the periodic structure of the
crystal is unperturbed, and therefore the corresponding stored
energy has to vanish. In other words, the nearest neighbors
interactions have to be labeled in the deformed configuration, and
not in the reference one. The rigorous way of formalizing this
idea is given exactly by the projection procedure introduced by
the operator $\pr$ in \eqref{SDd}.

\begin{oss}\label{ident} Here we describe the heuristic argument to identify
the $XY$ and $SD$ models just introduced. The correspondence
between the displacement functions $u_\e \in \asesd$ and the $XY$
fields $v_\e \in \asexy$ is given by identifying $u_\e$ with the
phase function of $v_\e$. More precisely, given a displacement
$u_\e \in \asesd$, the corresponding field $v(u_\e)$ for the $XY$
model is given by
\begin{equation}\label{sdtoxy}
v(u_\e)(l):= e^{2\pi i u_\e (l)} \qquad \text{ for every } l\in \oez.
\end{equation}
Viceversa, given $v_\e \in \asexy$, the corresponding displacement
is  $u(v_\e):= \tfrac{1}{2\pi}\theta_{v_\e}$, where
$\theta_{v_\e}\in [0,2\pi)$ is defined by the identity $v_\e = e^{i
\theta_{v_\e}}$. Note that the (arbitrary) choice of a precise
representative of the phase $\theta_{v_\e}$ of $v_\e$ does not
affect the elastic energy corresponding to $u(v_\e)$. Indeed, we
have
$$
2\pi \, \dist(u_\e(i)-u_\e(j),\Z)=d_{S^1}(v(u_\e)(i)-v(u_\e)(j)),
$$
where $d_{S^1}$ denotes the geodesic distance on $S^1$.
\par
Finally, we notice that by Taylor expansion we have
$2\pi(u_\e(i)-u_\e(j)) \approx v(u_\e)(i)- v(u_\e)(j)$, whenever
$u_\e(i)-u_\e(j)$ is small. Therefore, we expect  the
identification between the fields to produces small perturbations
in the corresponding energy densities (up to a pre-factor $4\pi^2$)  far from the singularities,
and this will be formalized in the sequel.
\end{oss}

\section{The topological singularities and energy functionals}\label{tosi}
As discovered by Jerrard in \cite{jerrard} (see also \cite{JS}), the Jacobian of
the order parameter is the  relevant geometric object carrying
energetic informations for the Ginzburg-Landau functionals. In the
same spirit, in this paper we introduce suitable discrete notions
for the topological singularities of the models we have introduced, and we consider them as the
meaningful variables of the corresponding discrete energy
functionals.
\par
For the $SD$ model, the topological singularity is given by the
discrete dislocation, defined through a discrete version of the
curl of the elastic strain, while for $XY$ spin systems, the
natural notion of topological singularity is that of discrete
vortex, that we will introduce in the sequel.
\par
 We will follow the
formalism introduced in \cite{ArOr} (see also \cite{Po})
representing the singularities by a discrete function $\alpha$
whose values on the cells of the lattice represent the topological
degree of the singularity. Moreover, we will identify this
function $\alpha$ with a measure concentrated on the center of the
cells of the lattice.
\par
As we will see in the sequel, we can always pass from a discrete
representation of the singularities to a continuous one (and
viceversa) by mean of interpolations procedures (and projection on
finite elements space respectively). This will also be
consistent with the topology we are going to use to perform the $\Gamma$-convergence analysis:
the distance, measured in  such a
topology, between the continuous representation of the
singularities and its discrete counterpart will turn out to be
vanishing for sequences of order parameters with bounded energy (see Proposition \ref{cdequi}).
\par

\subsection{The Jacobian}
Given $w\in H^1(\Om;\R^2)$, the Jacobian of $w$ is the $L^1$
function defined by
$$
J w:= \text{det} \nabla v.
$$
Let us denote by $C^{0,1}_c(\Om)$ the space of Lipschitz
continuous functions on $\Om$ with compact support, and by $X$ its
dual. The dual norm of $X$  will be denoted by $\|\cdot\|$.
\par
For every $w\in H^1 (\Om;\R^2)$, we can consider $J w$ as an
element of $X$ by setting
$$
<J w , \f>:= \int_{\Om} J w  \, \f \, dx \qquad \text{ for every }
\f\in C^{0,1}_c({\Om}).
$$
Note that $Jw$ can be written (in the sense of distributions) in a
divergence form as
\begin{equation}\label{div}
J w= \text{div } (w_1 (w_2)_{x_2}, - w_1 (w_2)_{x_1}),
\end{equation}
or equivalently,  in the form $J w = \curl (w_1 \nabla w_2)$. By a
density argument, we deduce that for every $\f \in C^{0,1}_c(\Om)$,
\begin{equation}\label{jwds}
<J w, \f>= - \int_{\R^2} w_1 (w_2)_{x_2} \f_{x_1} - w_1
(w_2)_{x_1} \f_{x_2} \, dx.
\end{equation}
Note that the right hand side of \eqref{jwds} is well defined also when 
$$
w\in
W^{1,1}(\Om;\R^2)\cap L^\infty(\Om;\R^2),
$$
 and for such a
function, we will take \eqref{jwds} as the definition of $Jw$ as
an element of $X$.

Finally, for later use we notice that for every $v:=(v_1, v_2), \,
w:=(w_1, w_2)$ belonging  to $H^{1}(\Om;\R^2)$ (or, as well, to
$W^{1,1}(\Om;\R^2)\cap L^\infty(\Om;\R^2)$), we have
\begin{equation}\label{for}
J v - J w = \frac{1}{2} \big(J (v_1 - w_1, v_2 + w_2) - J
(v_2-w_2,v_1+w_1)\big).
\end{equation}

By \eqref{jwds} and \eqref{for} we immediately deduce the
following lemma.

\begin{lemma}\label{cosu}
Let $v_n$ and $w_n$ be two sequences in $H^1(\Om;\R^2)$ such that
$$\|v_n - w_n\|_2 (\|\nabla v_n\|_2 + \|\nabla w_n\|_2) \to 0.
$$ Then $\|J v_n - J w_n\| \to 0 $ in $X$.
\end{lemma}

\subsection{Discrete dislocations}
Following the formalism introduced in \cite{ArOr} (see also
\cite{Po}), given a  function  $\xi: \oeu\to\R$ (playing the role
of a {\it pre-existing strain}),
we introduce its discrete curl $\dd\xi:\oed\to \R$,
defined for
every $ i\in \oed$ by
\begin{equation}\label{ddd}
\dd\xi (i) := \xi_{i,i + (0,\e)} + \xi_{i+(0,\e), i + (\e,\e)} -
\\
\xi_{i + (\e,0), i+(\e,\e)} -  \xi_{i,i+(\e,0)}. 
\end{equation}

Given an admissible displacement $u:\oez \to \R$, we recall that
we can decompose $\dd u$ in its elastic and plastic part (that is
optimal in energy), by setting $\beta_u^p = \pr \dd u$,
$\beta_u^e=\dd u - \beta^p_u$, where $\pr$ is defined in
\eqref{pr}.
\par We are now in a position to introduce the discrete dislocation
function $\alpha_u:\oed\to \{-1,\, 0,\, 1\}$, defined by
$$
\alpha_u (i):= \dd \beta_u^e (i) \qquad \text{ for every }
i\in\oed.
$$
It will be convenient also to represent $\alpha_u$ as a sum of
Dirac masses with weights in $\{-1,\, 0,\, 1\}$ and  supported on 
the centers of the squares where $\alpha \neq 0$,
setting
\begin{equation}\label{mudiu}
\mu_u:= \sum_{i\in \oed} \alpha_u (i)
\delta_{i+\frac{1}{2}(\e,\e)}.
\end{equation}

\begin{oss}
Notice that, in view of the very definition of $\pr$, we have that
$\beta_u^e\in (-1/2, \, 1/2]$, therefore by definition \eqref{ddd}
the discrete dislocation $\alpha_u$ takes values in $\{-1,\, 0,\,
1\}$. We deduce that, in our model, only singular dislocations can
be present in a cell.
\end{oss}

\subsection{Discrete vortices}
Given an admissible field $v \in\asexy$, the associated discrete
vorticity $\gamma_v:\oed\to \Z$ is defined by
\begin{equation}\label{vortdisc}
\gamma_v:=  \alpha_{\tfrac{1}{2\pi}\theta(v)},
\end{equation}
where $\theta(v)\in [0,2\pi)$ is the phase of $v$ defined by the
relation $v=e^{i \theta(v)}$. Moreover, we introduce the measure
$\mu_v$ defined by
\begin{equation}\label{mudiv}
\mu_v:= \sum_{i\in \oed} \gamma_v (i)
\delta_{i+\frac{1}{2}(\e,\e)}.
\end{equation}
Notice that,  as for the screw dislocations, in a cell we can have
 only  singular vortices.
\begin{oss}\label{delte}
Given $v \in \asexy$,  we can introduce a function $\tilde
v:\Om_\e\to\R^2$ that coincides with $v$ on $\oez\cap \Om_\e$ and
such that $J \tilde v  = \pi \mu_v$. Indeed,  consider
 the function $\tilde \theta(v)$ defined on each
segment $[i,j]$ with  $(i,j)\in\oeu$ by
$$
\tilde \theta(v) (i + s (j-i)) = \theta(v)(i) +  2\pi s \,
(\beta^e_{u(v)})_{i,j} \qquad \text{ for every } s\in[0,1],
$$
where $u(v):=\tfrac{1}{2\pi} \theta (v)$.
 Now,   extend $\tilde \theta(v)$ in each cell $i +
[0,\e]^2$ (with $i\in\oed$), making it zero-homogeneous with
respect to the center of the cell. Finally, on each cell we set
$\tilde v:= e^{i \tilde \theta(v)}$.
\par
A straightforward computation leads to the equality $J \tilde v  =
\pi \mu_v$.  This argument shows that the vorticity function
$\alpha_v$ represents a discrete version of the Jacobian, and
seems a very natural object in this context.
\par
Finally, for latter use, we notice that it can be easily proved
(for instance by the identity $J\tilde v= \curl (\tilde v_1 \nabla
\tilde v_2)$) that for each $\e$-square $Q_i$ with $i\in\oed$ we
have
\begin{equation}\label{sto}
\mu_v (Q_i) = \frac{1}{\pi}\int_{\partial Q_i} \tilde v_1
\frac{\partial}{\partial s} \tilde v_2 \, ds.
\end{equation}
\end{oss}


\begin{oss}\label{estifo}
We notice here that a very easy estimate  leads to the following
bound for the  total variation of the topological singularities
$$
|Jw_\e|(\Om)\le C \, GL_\e(w_\e) , \quad |\mu_{v_\e}|(\Om) \le C
\,  XY_\e(v_\e), \quad |\mu_{u_\e}|(\Om) \le C \, SD_\e(u_\e).
$$
\end{oss}
\subsection{The energy functionals}
In this paragraph we will formally introduce the (rescaled) energy
functionals corresponding to the $XY$ spin system, the screw dislocations
and the Ginzburg-Landau models written in terms of the
corresponding singularities.
\par
In order to rewrite the energy functionals in these new variables,  we minimize the free
energies among all field quantities compatible with the prescribed
singularity. For instance, in the $SD$ model, we will fix
$\mu:=\curl \beta^e$, and   minimize the elastic energy among all
$u$  compatible with $\mu$, {\it i.e.}, with $\mu_u=\mu$. The
corresponding energy functional can be thought of as the energy
stored in the crystal, for a certain
 given dislocation density.
\par
We will consider the energetic regimes of order $|\log\e|^h$, where
$h$ is any fixed positive real number. Following the convention
according to which the infimum of the empty set is $+\infty$, we
define the Ginzburg-Landau energy functional $\GL_\e : X \to [0,+\infty]$,
where $X=(C^{0,1}_c(\Om))^*$, as
\begin{equation}\label{glene}
\GL_\e(\mu):= \frac{1}{|\log \e|^h} \inf \left\{  GL_\e(w), 
w\in H^1({\Om};\R^2): \frac{J(w)}{\pi|\log\e|^{h-1}}
 =\mu \right\}.
\end{equation}

Let us pass to the energy functionals corresponding to $XY$ spin
systems. Using the notations introduced in the Section \ref{tosi},
the energy functionals $\XY_\e: X \to [0,+\infty]$ are defined by
\begin{equation}\label{xyene}
\XY_\e(\mu):= \frac{1}{|\log \e|^h} \inf \left\{ XY_\e(v), \, v\in
\asexy: \frac{\mu_v}{|\log\e|^{h-1}}  =\mu  \right\}.
\end{equation}

Finally  the energy functionals corresponding to the screw
dislocations model $\SD_\e: X \to [0,+\infty]$ are defined by
\begin{equation}\label{sdene}
\SD_\e(\mu):= \frac{4 \pi^2}{ |\log \e|^h} \inf \left\{ SD_\e(u),
\, u\in \asesd: \frac{\mu_u}{|\log\e|^{h-1}}  =\mu  \right\},
\end{equation}
where the prefactor $4\pi^2$ is just a normalization factor which guarantees that the family $(\SD_\e)$
asymptotically behaves as $\XY_\e$ and $\GL_\e$.

\section{The variational equivalence argument}\label{vargu}
In this section we will introduce a notion of equivalence between
families of functionals defined on  a metric space $(X,d)$,
depending on a small parameter $\e$. Such a notion turns out to be efficient to 
compare different variational models which share the same asymptotic behavior as $\e$ goes to zero.

\subsection{The notion of variational equivalence}
Let $(F_\e)$ and $(G_\e)$ be two families of functionals from $X$ to
$\R\cup\{\infty\}$ depending on the  parameter $\e \in \R^+ \cup \{0\}$.

\begin{definition}\label{equivalence}
{\rm

We set $(G_\e) \preceq (F_\e)$ if there exists a continuous
increasing function $\e \mapsto \delta_\e$, with $\delta_0=0$,
such that the following holds.
\par
For every $\e_n\to 0$, and  $(p_n)\subset X$ such that $F_{\en} (p_n)\le C$,
there exists  a family $(q_n) \subset X$ such that
\begin{itemize}
\item[i)] $\limsup_n (G_{\delta_{\en}}(q_n) - F_{\en} (p_n))\le
0$; 
\item[ii)] Either $(p_n)$ and $(q_n)$ are unbounded or
$d(p_n , q_n) \to 0$  as $n \to +\infty$.
\end{itemize}

We set $(F_\e) \simeq (G_\e)$, and we say that  $(F_\e)$ and $(G_\e)$ are variationally equivalent (for $\e \to 0$) if $(F_\e) \preceq (G_\e)$ and $(G_\e) \preceq (F_\e)$.}
\end{definition}

Note that the relation $\preceq$ just introduced is transitive, {\it i.e.}, if $F^1_\e \preceq F^2_\e$ and $F^2_\e\preceq F^3_\e$ then $F^1_\e \preceq F^3_\e$.
Moreover the relation $\simeq$  is  an equivalence relation between families of functionals $\{(F_\e), F_\e:X\to \R\}$, {\it i.e.}, the following properties hold
\begin{align*}
& \text{Reflexivity:} & (F_\e) \simeq (F_\e);\\
& \text{Symmetry:} & (F_\e) \simeq (G_\e)  \text{ implies } (G_\e) \simeq (F_\e);\\
& \text{Transitivity:} & (F_\e) \simeq (G_\e) \text{ and } (G_\e) \simeq (H_\e)
\text{ imply } (G_\e) \simeq (H_\e).
\end{align*}
\subsection{Some consequences of the variational equivalence}

Let us consider now a first important relation between the notion of variational equivalence and that of $\Gamma$-convergence (for the definition and the main properties of $\Gamma$-convergence we refer the reader to \cite{braidesbook} and \cite{dm}). To this purpose we recall the following definition of equi-coercivity: a family of functionals $(F_\e)$ is said to be equi-coercive if, given $\en\to 0$ and $(p_n)$ such that $F_{\en}(p_n)<C$, then $(p_n)$ is relatively compact in $X$.
\begin{theorem}\label{ordiga}
Let  $(F_\e)$ be a family of equi-coercive functionals
 $\Gamma$-converging to some functional $H$ in $X$.
Then $(F_\e)$ is  variationally equivalent to the constant sequence
$H_\e \equiv H$.
\end{theorem}

\begin{proof}
We begin
proving $F_\e \preceq H$ by a contradiction argument. Assume by
contradiction that the relation $F_\e \preceq H$ does not hold;
then there exist sequences $\en\to 0$,  $(p_n) \subset X$ with
$H(p_n) \le C$, and $r > 0$ such that
\begin{equation}\label{contr1}
\inf_{x\in B_r(p_n)} F_{\e_n}(x) \ge H(p_n) + r.
\end{equation}
On the other hand, by the equi-coercivity of $F_\e$ we easily
deduce the coercivity of $H$, so that we can assume without loss
of generality that $p_n$ converges to some $p\in X$. By the
$\Gamma$-limsup inequality there exists a sequence $(z_n)\subset
X$ such that
\begin{equation}\label{contr2}
z_n \to p, \qquad F_{\e_{n}}(z_n) \to H(p) \le \liminf H(p_n).
\end{equation}
Equation \eqref{contr1} together with \eqref{contr2} provides a
contradiction.
\par
Let us pass to the proof of the opposite inequality $H\preceq F_\e$.
Assume by contradiction that there exist  sequences $\en \to 0$,
$(p_n)\subset X$ with $F_{\en}(p_n)\le C$ and $r>0$ such that

\begin{equation}\label{contr3}
\inf_{x\in B_r(p_n)} H(x) \ge F_{\en}(p_n) + r.
\end{equation}

By the equi-coercivity property of $F_\e$ we may assume that $p_n
\to p$ in X, and by  the $\Gamma$-liminf inequality we have that,
for n big enough
$$\inf_{x\in B_r(p_n)} H(x) \le H(p) \le \liminf F_{\en}(p_n),$$
 which, together with \eqref{contr3} provides a contradiction.
\end{proof}

\begin{oss}

Note that in Theorem \ref{ordiga} the equicoercivity assumption
plays a fundamental role. Indeed consider the space $l_2$ of
sequences $\{a_i\}$ with $\sum_i a_i^2$ finite, and consider the
functionals $F_{\e}:\R\to \R$ defined by
$$
F_\e(\{a_i\}):= \sum_{i\neq[1/\e]} a_i^2,
$$
where $[1/\e]$ denotes the integer part of $1/\e$. Clearly  $F_\e$
$\Gamma$-converge to the functional $H(\{a_i\}):=\|\{a_i\}\|_2^2$,
but $F_\e$ is not variationally equivalent to $H$.
\end{oss}

The next theorem, together with Theorem \ref{ordiga}, clarifies the
relation between the notion of $\Gamma$-convergence and the notion
of variational equivalence.
\begin{theorem}\label{consequi}
Let $(F_\e)$ and $(G_\e)$ be variationally equivalent. Then the
following properties hold.
\begin{itemize}
\item[1)]
$F_\e$ are equi-coercive if and only if $G_\e$ are equi-coercive;
\item[2)]
$F_\e$ $\Gamma$-converge to some functional $H$ in $X$ if and only
if $G_\e$ $\Gamma$-converge to $H$.
\end{itemize}
\end{theorem}

\begin{proof}
To prove property $1)$ assume for instance that $G_\e$ are
equi-coercive and let us prove that so are also $F_\e$. To this
aim let $\e\mapsto \delta_\e$ be the map given in the definition of $G_\e\preceq F_\e$. 
Let $\en\to 0$ and $(p_n)\subset X$ be such that $F_{\en}(p_n)\le C$. By Definition
\ref{equivalence} there exists a sequence $(q_n)\subset X$ such
that $\limsup_n G_{\delta_{\en}} (q_n) \le C$. By the equi-coercivity
property of $G_\e$ we deduce that (up to a subsequence) $q_n \to
x$ for some $x\in X$. By property $ii)$ of Definition
\ref{equivalence} we deduce that $d(p_n , q_n)\to 0$ and
therefore that  $p_n\to x$.
\par
Let us pass to property 2). Assume for instance that $F_\e$
$\Gamma$-converge to $H$ and let us prove that also $G_\e$
$\Gamma$-converge to  $H$.
To this purpose, let $\en\to 0$ and let $\rho_n$ be such that 
$\en=\delta_{\rho_n}$. In order to prove the $\Gamma$-limsup inequality let $p\in X$,
and let $(p_n)$ be a recovery sequence for $F_{\rho_n}$; {\it i.e.}, $p_n\to p$ and $F_{\rho_n}(p_n)\to F(p)$. Consider the
sequences $(q_n)\subset X$ given by Definition
\ref{equivalence}. By property ii) of Definition \ref{equivalence}
we have that $q_n\to p$, and by property i) we have that
$$
\limsup G_{\en} (q_n) \le \lim_n F_{\rho_n}(p_n) = H(p).
$$
In order to prove the $\Gamma$-liminf inequality, let now $\e\to\de$ be the map 
given in the definition of $F_\e\preceq G_\e$ and let $p _n \to p$. 
By Definition \ref{equivalence} there exists a sequence
$q_n\to p$ such that
$$
H(p)\le \liminf_n F_{\delta_{\en}}(q_n) \le \liminf_n G_{\en}(
p_n),
$$
that proves the $\Gamma$-liminf inequality.
\end{proof}

\begin{oss}
Note that if $F_\e$ and $G_\e$ are variationally equivalent, by
Theorem \ref{consequi} we deduce that they share the same
$\Gamma$-limits even when they are not  equi-coercive.
\par
We also notice that in the class of equi-coercive functionals
admitting a $\Gamma$-limit, our definition of equivalence
coincides with that given (at order zero) in \cite{BrTr}.
\end{oss}

\begin{oss}
Even if Definition \ref{equivalence} is given in a metric context,
it can be generalized to vectorial topological spaces.  We
chose the metric framework because it is usually general enough
for practical applications, the metric $d$ being either the metric
inducing the weak topology of a ball in a Banach space, or (in
view of a compact embedding) the distance induced by the norm of a
Banach space.
\end{oss}

\subsection{A first example of equivalent families}

For all positive $s>0$ consider the Ginzburg-Landau functionals $GL^s_\e:W^{1,2}(\Om;\R^2)\to [0,+\infty)$,
 defined as

\begin{equation}\label{GLs}
GL^s_\e(w):= \frac{1}{2} \int_{{\R^2}} |\nabla w|^2 + \frac{s}{\e^2} \W(w).
\end{equation}
Note that for $s=1$ the functionals $GL^s_\e$ coincides with the functional $GL_\e$ defined in \eqref{GL}.
Consider also the functional $\GL_\e^s:X\to [0,+\infty]$ defined as in \eqref{glene} with $GL_\e$ replaced by $GL^s_\e$.

\begin{proposition}\label{s1s2}
For every $0<s_1<s_2$, the families of functionals $\GL_\e^{s_1}$ and $\GL_\e^{s_2}$ are variationally equivalent, according to Definition \ref{equivalence}.
\end{proposition}

\begin{proof}
Let $s_1< s_2$. Since $ \GL_\e^{s_1} \le \GL_\e^{s_2}$, we immediately deduce
$ \GL_\e^{s_1} \preceq \GL_\e^{s_2}$. In order to prove the opposite relation, consider the following change of variables
$$
\e {\mapsto} \de:=   \e \left(\frac{s_2}{ s_1}\right)^{\frac{1}{2}}.
$$
Following Definition \ref{equivalence}, let $\e_n\to 0$ as $n\to +\infty$ and let  $(p_n)\subset X$ be
such that $\GL_{\e_n}^{s_1}(p_n)\le C$. Moreover,  set
$$t_n:=
\tfrac{|\log {\delta_{\e_n}}|^{h-1}}{|\log \e_n|^{h-1}}, \qquad  q_n:=
\frac{p_n}{t_n}. $$
 By the fact that $t_n\to 1$ as $n\to +\infty$ it
immediately follows that either $p_n$ and $q_n$ are both
unbounded or $d(p_n,q_n) \to 0$.
\par
 Let   $(w_n)\subset H^1({\R^2};\R^2)$ be such that
 $$
 \frac{1}{|\log\e_n|^{h-1}} J w_n= p_n, \qquad \tfrac{1}{|\log \e_n|^h} GL^{s_1}_{\e_n}(w_n) - \GL^{s_1}_{\e_n}(p_n) \to 0  \text{ as } n\to +\infty.
 $$
 Then, since by construction
$\tfrac{1}{|\log {\delta_{\e_n}}|^{h-1}} J w_n=q_n$,
 we have
\begin{multline*}
\limsup_n (\GL_{\delta_{\e_n}}^{s_2}(q_n) - \GL_{\e_n}^{s_1}(p_n))\\ \le
\limsup_{n} \left(\frac{1}{|\log{\delta_{\e_n}}|^h} GL_{\delta_{\e_n}}^{s_2}(w_n)
- \frac{1}{|\log{\e_n}|^h} GL_{\e_n}^{s_1}(w_n)\right)\\
= \limsup_{n} \Big( \frac{1}{|\log \left(\e_n \frac{s_2}{s_1}
\right)^{\frac{1}{2}}|^h} - \frac{1}{|\log \e_n|^h}\Big) \left(
\frac{1}{2} \int_{{\R^2}} |\nabla w_n|^2 + \frac{s_1}{\e_n^2}
\W(w_n) \right)=0.
\end{multline*}
\end{proof}

In the following we will need a slight variant of Proposition \ref{s1s2}, where the pre-factor $s$ in front of the potential term $W$ is a suitable function of $\e$.

\begin{proposition} \label{se}
Let $\e \mapsto s_\e$ be an increasing functions from $\R^+$ to $\R^+$ 
such that  $\tfrac{|\log s_\e|}{|\log\e|}\to 0$ as $\e \to 0$.
Then the functionals $\GL_\e$ are variationally equivalent to
$\GL^{s_\e}_\e$, according to Definition \ref{equivalence}.
\end{proposition}

\begin{proof}
The proof follows the lines of the proof of  Proposition \ref{s1s2},
setting now $\de:= \tfrac{\e}{\sqrt{s_\e}}$.
\end{proof}

\section{Variational equivalence between $\GL_\e$, $\XY_\e$ and
$\SD_\e$}\label{mainsec}
In this Section we prove the main result of
the paper: the energies corresponding to $XY$
spin systems,  to screw dislocations and to the Ginzburg-Landau
model are variationally equivalent in the sense of
Definition \ref{equivalence}. We prove this result for any scaling
of the energies of order $|\log\e|^h$ with $h\ge 1$, the most relevant cases being  $h=1$ and $h=2$.
\par
In order to define all the energy functionals in the same space,
we consider the Banach space $(X,\|\cdot\|)$ defined as the dual
of Lipschitz functions with compact support in $\Om$,  endowed
with the dual norm. This space is the natural space for the
topological singularities, that are the relevant quantities in all
the investigated models. The functionals $\GL_\e$, $\XY_\e$ and
$\SD_\e$ from $X$ to $[0,+\infty]$ are defined rigorously in
\eqref{glene}, \eqref{xyene} and \eqref{sdene} respectively.

\begin{theorem}\label{mainthm}
The functionals $\XY_\e$, $\SD_\e$ and $\GL_\e$ are variationally
equivalent according to Definition \ref{equivalence}.
\end{theorem}

We will prove in three steps that $\XY_\e \preceq \GL_\e\preceq
\SD_\e \preceq \XY_\e$. This relations imply the equivalence of
all the functionals by the transitivity property of the order
relation $F \preceq G$. Before the proof of Theorem \ref{mainthm}
we state and prove a proposition that will be used to identify the
discrete description of the singularities ({\it i.e.}, the discrete
measures $\mu_u$ and $\mu_v$ defined in \eqref{mudiu} and
\eqref{mudiv}, respectively) with the diffuse Jacobians of some
suitable interpolation of the corresponding fields. To this
purpose, for every positive $\e>0$, consider the triangulation
 $\ts_\e:= \{ T^\pm(i), \, i\in\e \Z^2\}$  whose
triangles are of the type
\begin{align}\label{triangoli}
& T^-(i):=\{(x_1,x_2): i_1\le x_1\le i_1+\e, \, i_2\le x_2 \le i_2
+ x_1\};\cr & T^+(i):=\overline{(i_1, i_1 + \e)\times (i_2, i_2 +
\e) \setminus T^-_\e(i)}.
\end{align}
To any  $v_\e$  in $\asexy$, we can associate the continuous field
$w(v_\e)$ in $H^1(\Om_\e;\R^2)$   given by the piece-wise linear
interpolation of $v_\e$ on the triangles of $\ts_\e$ contained in
$\Om_\e$. Note that the Jacobian of $w(v_\e)$ is piecewise constant
and that it can be extended by zero to a $L^1$ function in $\Om$.
Therefore $J(w(v_\e))$ can be seen as an element of $X$.
\par
Since we will need to localize the energy functionals to subsets
of $\Om$, we introduce the following notation. Given $A\subset
\Om$, we will denote by $GL_\e(\cdot,A)$ the restriction of the
$GL_\e$ energy density to $A$, and we denote by $XY_\e(\cdot,A)$
and $SD_\e(\cdot,A)$ the restriction of the corresponding energy
densities to the nearest neighbors contained in $A$. Finally, for
every given positive $\delta$ we set
\begin{equation}\label{idelta} I_\delta(\Om):= \{x\in \Om:
\min_{y\in\partial\Om} \min_{i\in\{1,2\}}|x_i - y_i|> \delta\}.
\end{equation}

\begin{proposition}\label{cdequi}
 Let $v_\e \in \asexy$ be a sequence with
$XY_\e(v_\e)\le C |\log\e|^h$. Then $\|\frac{J w(v_\e)}{\pi} -
\mu_{v_\e}\| \to 0$ as $\e\to0$.
\end{proposition}

\proof Let $\de:= \left( \e |\log \e|^{h+1} \right)$,  denote by
$\tilde \de $ the projection of $\de$ on $\e\Z$, and set
\begin{align*}
& R^h_{\tilde\de}:= \{(x,y) \in \R^2: y\in \tilde\de\Z \}, \\
& R^v_{\tilde\de}:= \{(x,y) \in \R^2: x\in \tilde\de\Z \},\\
& R_{\tilde\de}:=( R^h_{\tilde\de} \cup R^v_{\tilde\de}) \cap
I_{\tilde\de}(\Om).
\end{align*}
If we denote by $R^1_{\tilde\de}(s)$ the class of nearest
neighbors in $ R_{\tilde\de} + (s,s)$, by the Mean Value Theorem
it is easy to prove that  there exists $s_{\tilde \de}\in \e\{0,1,
\ldots, \tfrac{\tilde \de}{\e}\}$ such that
\begin{align}\label{fuxyd}
\nonumber
XY_{\e}(v_\e) \ge & \sum\limits_{s\in \e\{0,1, \ldots, \tfrac{\tilde
\de}{\e}\}}  \sum_{l,k\in R_{\tilde\de}^1(s)}
\frac{1}{2} |v_\e(l) - v_\e(k)|^2  \\ 
 \ge & \frac{\tilde \de}{\e}
\sum_{l,k\in R_{\tilde\de}^1(s_{\tilde \de})} \frac{1}{2} |v_\e(l)
- v_\e(k)|^2.
\end{align}

By \eqref{fuxyd} we deduce that
$$
\max_{l,k\in R_{\tilde\de}^1(s_{\tilde \de})}  |v_\e(l) - v_\e(k)|
\to 0 \qquad \text{ as } \e \to 0.
$$

Let us denote by  $Q_{\tilde\de}^i$ the cubes contained in
$I_{\tilde\de}(\Om) + (s_{\tilde \de},s_{\tilde \de})$ of the type
$$Q_{\tilde\de}^i:=
\tilde\de (i+[0,1]^2) +(s_{\tilde \de},s_{\tilde
\de}), \, i\in \Z^2,
$$
and by $A_{\tilde \de}$ their union.

Let $\f\in C^{0,1}_c(\Om)$ with norm less then one, and denote by
$\tilde\f$ the locally constant function that on each
$\tilde\de$-square $Q_{\tilde\de}^i$ coincides with $\f$ on the
center of $Q_{\tilde\de}^i$. Then we have

\begin{multline}\label{stimona} |<\frac{J w(v_\e)}{\pi} - \mu_{v_\e}, \f>| \\
 \le
\left(\int_{\Om_\e \setminus A_{\tilde \de}} |J w(v_\e)| dx +
|\mu_{v_\e}|(\Om_\e \setminus A_{\tilde \de})\right) \max_{\Om_\e \setminus A_{\tilde \de}} |\f|\\
+ \sum_{Q_{\tilde\de}^i\subseteq A_{\tilde \de}} \left(
\int_{Q_{\tilde \de}^i} |J(w(v_\e))| dx + |\mu_{v_\e}|(Q_{\tilde
\de}^i) \right)|\f - \tilde \f| \\
+ \sum_{Q_{\tilde\de}^i\subseteq
A_{\tilde \de}} \left( \int_{Q_{\tilde \de}^i}   \frac{J
w(v_\e)}{\pi} dx -\mu_{v_\e}(Q_{\tilde \de}^i) \right) \tilde \f.
\end{multline}
The first two addends of the right hand side of \eqref{stimona} are
vanishing (uniformly with respect to $\f$ belonging to
$C^{0,1}_c(\Om)$ and with norm less then one), since $\max_{\Om_\e
\setminus A_{\tilde \de}} |\f|$ and $|\f - \tilde \f|$ are bounded
by $C\tilde \de$, and (see Remark \ref{estifo})
$$
\int_{\Om_\e} |J w(v_\e)| dx + |\mu_{v_\e}|(\Om_\e) \le C
|\log\e|^h.
$$
Therefore it remains to prove that also the third addend in
\eqref{stimona} is vanishing, uniformly with respect to $\f$. To
this purpose, let $\tilde v_\e$ be defined as in Remark
\ref{delte} on the boundaries of all $Q_{\tilde\de}^i$. Since
$J(w(v_\e))= \curl (w(v_\e))_1 \nabla (w(v_\e))_2$ and $ J\tilde
v_\e= \curl (\tilde v_\e)_1 \nabla (\tilde v_\e)_2$ (see also of
\eqref{sto}), we have
\begin{multline}\label{ide}
\int_{Q_{\tilde \de}^i} \frac{J w(v_\e)}{\pi}  dx
-\mu_{v_\e}(Q_{\tilde \de}^i) \\
= \frac{1}{\pi}\int_{\partial
Q_{\tilde \de}^i} (w(v_\e))_1 \frac{\partial}{\partial s}
(w(v_\e))_2 - (\tilde v_\e)_1
\frac{\partial}{\partial s} (\tilde v_\e)_2\\
= \frac{1}{2 \pi}\int_{\partial Q_{\tilde \de}^i} \big(w(v_\e) -
\tilde v_\e \big)_1 \frac{\partial}{\partial s} \big(w(v_\e) +
\tilde v_\e\big)_2
\\
 - \big(w(v_\e) - \tilde v_\e \big)_2
\frac{\partial}{\partial s} \big(w(v_\e) + \tilde v_\e\big)_1.
\end{multline}

By  \eqref{fuxyd} and \eqref{ide} we conclude

\begin{eqnarray*}
\sum_{Q_{\tilde\de}^i} \left( \int_{Q_{\tilde \de}^i} \frac{J w(v_\e)}{\pi} dx -\mu_{v_\e}(Q_{\tilde \de}^i) \right) \tilde \f 
\le C \sum_{Q_{\tilde\de}^i}  \int_{\partial Q_{\tilde \de}^i}
|w(v_\e)- \tilde v_\e| |w(v_\e)'|\\ \le C
\left(\sum_{Q_{\tilde\de}^i}  \int_{\partial Q_{\tilde \de}^i}
|w(v_\e) - \tilde v_\e|^2  \sum_{Q_{\tilde\de}^i} \int_{\partial
Q_{\tilde \de}^i} |w(v_\e)'|^2\right)^{\frac{1}{2}}\\ \le 
 C
\sum_{l,k\in R_{\tilde\de}^1(s_{\tilde \de})} |v_\e(l) -
v_\e(k)|^2  \le C \frac{1}{|\log\e|}\to 0,
\end{eqnarray*}

where  in the last but one inequality we have used that $|w(v_\e)
- \tilde v_\e|$ is controlled on the segment $[l,k]$ by  $|v_\e(l)
- v_\e(k)|$.


\begin{oss}\label{id2}
Note that, given $w(v_\e)$, we can always extend it to a function
$\tilde w(v_\e)$ defined in the whole $\Om$ such that 
$$\|\nabla
\tilde w(v_\e)\|_{L^2(\Om;\R^{2\times 2})} \le C \|\nabla
w(v_\e)\|_{L^2(\Om_\e;\R^{2\times 2})}.$$
 As a consequence, if
$w(v_\e)$ is with finite energy as in Proposition \ref{cdequi}, then it
is easy to see that $\|J w(v_\e) - J \tilde w(v_\e)\|\to 0$ as
$\e\to 0$. Therefore, the conclusion of Proposition \ref{cdequi} still
holds with $J \tilde w(v_\e)$ in place of $J w(v_\e)$.
\end{oss}

We are now in a position to prove Theorem \ref{mainthm}. By simplicity of notation 
in what follows we will replace $\e_n$ by $\e$.


\subsection{Proof of $\XY_\e \preceq \GL_\e$}

In order to follow Definition \ref{equivalence} we have first to define the function $\e\mapsto
\de$.  A convenient choice for this purpose is to set $ \tilde
\de:= \left( \e |\log \e|^{h+1} \right)$, and $\de:=\lambda_\e
\tilde \de$, where a suitable factor $\lambda_\e \to 1$ will be
chosen in the following. Let $(\mu_\e)\subset X$ be a sequence
such that $\GL_{\e}(\mu_\e)\le C$.  By the very definition of
$\GL_{\e}$ there exists $(w_\e)\subset H^1(\Om;\R^2)$ such that
$$
\frac{1}{|\log\e|^{h-1}} \frac{J w_\e}{\pi} = \mu_\e  \quad \text{
and } \quad \frac{1}{|\log\e|^h}GL_{\e}(w_\e) - \GL_{\e}(\mu_\e)
\to 0 \quad \text{ as } \e\to 0.
$$
let $I_{\tilde\de}(\Om)$ be defined as in \eqref{idelta} (with
$\delta$ replaced by $\tilde\delta_\e$), and consider the nets
$R^h_{\tilde\de}$, $R^v_{\tilde\de}$ and $R_{\tilde\de}$ defined
by

\begin{align*}
& R^h_{\tilde\de}:=  \{(x,y) \in \R^2: y\in \tilde\de\Z \}, \\
& R^v_{\tilde\de}:=  \{(x,y) \in \R^2: x\in \tilde\de\Z \},\\
& R_{\tilde\de}:=  ( R^h_{\tilde\de} \cup R^v_{\tilde\de}) \cap I_{\tilde\de}(\Om).
\end{align*}

By the Mean Value Theorem, it is easy to prove that for every
$\tilde \de$ we can find $s_{\tilde \de}\in (0,\tilde\de)^2$ such
that
\begin{equation}\label{netdis}
{\tilde\de} \int_{R_{\tilde\de} + s_{\tilde\de} } \frac{1}{2} |
w'_\e|^2 + \frac{1}{ 2\e^2} W ( w_\e) \le
 \int _{\Om} \frac{1}{2} |\nabla  w_\e|^2 + \frac{1}{\e^2} W( w_\e),
\end{equation}
where $ w'_\e (x)$ denotes the tangential derivative  of $ w_\e
\res ( R_\de + s_{\tilde\de})$ in $x$  (that is well defined by
standard slicing arguments).
\par
{\bf Claim:} let $L$ be a segment with length larger than $\e$ and
let $\f\in H^1(L;\R^2)$, then
\begin{equation}\label{claim}
\int_{L} \frac{1}{2} |\f'|^2 + \frac{1}{2 \e^2} W (\f) \ge  C
\frac{ \max_{t\in L} (|\f| - 1)^2}{\e}.
\end{equation}
We now use the Claim, that we will prove later, in order to conclude the proof of $\XY_\e \preceq \GL_\e$.
\par
Set $\tilde v_{\e}$ and  $\hat v_{\e}$ from
$R_{\tilde\de}$ to $\R^2$  as follows
$$
\tilde v_{\e}(l):=  w_\e (l+ s_{\tilde\de}), \quad
 \hat v_{\e}(l):=
\frac{\tilde v_{\e}(l)}{|\tilde  v_{\e}(l)|} \qquad
\text{for every } l\in {\tilde\de} \Z^2\cap I_{\tilde\de}(\Om).
$$
Note that by the Claim  \eqref{claim} and by the choice of
$\tilde\de$ we immediately deduce that
$$
|\tilde v_\e| - 1 \to 0 \qquad  \text{ uniformly  }  \text{ as }
\e\to 0.
$$
Here we use the assumption that $\Om$ is star-shaped with respect
to the origin.  Let $d$ be the distance between $0$ and $\partial
\Om$. We set
\begin{equation} \label{del}
\lambda_\e:=\tfrac{d}{d-2\tilde\de},
\end{equation}
so that $\Om\subseteq\lambda_\e I_{2 \tilde \de}(\Om)$ and $  \Om_{\de}^0 \subseteq \lambda_\e\big(\tilde\de \Z^2\cap I_{\tilde\de}(\Om)\big)$ . We are in a
position to introduce the function
\begin{equation}
v_{\e}(l):= \hat v_{\e}(\frac{l}{\lambda_\e}) \qquad
\text{ for every } l\in \Om_{\de}^0.
\end{equation}
By Jensen inequality and in view of \eqref{netdis} we have
\begin{align}\label{dis1}
XY_{\de}(v_\e) &= \frac{1}{2} \sum_{(i,j)\in\Om_\de^1}  |v_\e( i)
- v_\e( j)|^2 \\ &= 
 (1 + o(1))\sum_{(i,j)\in (I_{\tilde\de}(\Om))^1_{\tilde\de}} \frac{1}{2} |\tilde v_\e( i) - \tilde v_\e( j)|^2 \nonumber \\ &=
(1+ o(1)) \sum_{(i,j)\in (I_{\tilde\de}(\Om))^1_{\tilde\de}}
\frac{1}{2} | w_\e( i + s_{\tilde\de}) -  w_\e( j +
s_{\tilde\de})|^2  
\nonumber \\ &\le  \de (1+ o(1)) \int_{(R_\de + s_\de ) }
\frac{1}{2} | w'_\e|^2 
\nonumber \\ &\le
 (1+ o(1)) \int _{\Om } \frac{1}{2} |\nabla  w_\e|^2 +
\frac{1}{\e^2} W( w_\e) = (1+o(1)) GL_{\e}(w_\e),\nonumber
\end{align}
where we recall that $(I_{\tilde\de}(\Om))^1_{\tilde\de}$ denotes
the pairs $(i,j) \in \tilde\de \Z^2\cap I_{\tilde\de}(\Om)$ with
$i < j$ and $|i-j|=1$.
\par
We are in a position to introduce the sequence of variables
$(\eta_\e)\subset X$ for the $\XY_\de$ functionals, satisfying
properties i) and ii) of Definition \ref{equivalence},
$$
\eta_\e:= \frac{1}{|\log\de|^{h-1}} \mu_{v_\e}.
$$
By \eqref{dis1}, since $\tfrac{|\log \e|}{|\log\de|} \to 1$, we
deduce that property i) of  Definition \ref{equivalence} is
satisfied.
\par
Now we will prove that  $\|\mu_\e - t_\e \eta_\e\| \to 0$ for
some $t_\e\to 1$, that will ensure property ii) of Definition
\ref{equivalence}. To this purpose, set
$$ t_\e:=
\tfrac{|\log\de|^{h-1}}{|\log\e|^{h-1}}, \quad \text{ so  that }
\quad \mu_\e - t_\e \eta_\e= \tfrac{1}{|\log\e|^{h-1}} \left(
\frac{J w_\e}{\pi} - \mu_{v_\e} \right). $$

In view of Proposition \ref{cdequi}, to conclude it is enough to
show that
\begin{equation}\label{doid}
\tfrac{1}{|\log\e|^{h-1}}  \|J w_\e - Jw(v_\e) \| \to 0 \quad
\text{ as } \e\to 0.
\end{equation}

Note that (see Remark \ref{id2}), we can always extend $w(v_\e)$
to $\Om$ (and we will still denote this extension by $w(v_\e)$)
such that $\|\nabla w(v_\e)\|^2_2  \le C |\log \e|^h$. Therefore,
by Lemma \ref{cosu} and Remark \ref{id2} (since we  also have
$\|\nabla w_\e\|^2_2 \le C |\log \e|^h$), in order to prove
\eqref{doid} it is enough to check that
$$
|\log\e|^h \|w(v_\e) - w_\e\|^2_2 \to 0 \qquad \text{ as } \e\to
0.
$$
Let $I_{2\de}(\Om)$ be defined as in \eqref{idelta} with $\de$
replaced by $2\de$ (so that all the functions we have just
introduced are defined on $I_{2\de}(\Om)$). Using that
$|\Om\setminus I_{2\de}| \le C \de$ and that the potential $W$ in
the $GL_\e$ functionals controls the $L^2$ norm of $w_\e$, it is
easy to prove that
$$|\log\e|^h \|w(v_\e) -
w_\e\|^2_{L^2(\Om\setminus I_{2\de}(\Om);\R^2)} \to 0.
$$
Therefore, we will estimate the $L^2$ norm only on
$I_{2\de}(\Om)$.

By triangular inequality we have
\begin{multline}\label{bo}
\int_{I_{2\de}(\Om)}   |w(v_\e) - w_\e|^2 \le C
\int_{I_{2\de}(\Om)}  |w(v_\e) -
w(\hat v_{\e})|^2 +
| w(\hat v_{\e}) - w(\tilde v_{\e})|^2  \\ + C \int_{I_{2\de}(\Om)} | w( \tilde
v_{\e}) - w_\e(x+s_{\tilde\de})|^2 + | w_\e(x+
s_{\tilde\de}) - w_\e|^2. 
\end{multline}

We can easily estimate the first and the last addend in the
right-hand side of \eqref{bo} as follows
\begin{multline*}
\int_{I_{2\de}(\Om)} \Big( |w(v_\e) - w(\hat v_\e)|^2 +
| w_\e(x + s_{\tilde\de}) - w_\e|^2\Big) \, dx \\
\le C \de^2
\int_{\Om} |\nabla w_\e|^2 \, dx \le C \de^2 |\log \e|^h.
\end{multline*}

Let us pass to the second addend. For any $(i,j) \in
(I_{\tilde\de}(\Om))^1_{\tilde\de}$, set  $L_{i,j}:=s_{\tilde\de}+
[ i, j]$. Moreover set $m_{i,j}:= \max_{t\in L_{i,j}}
\big||w_\e(t)| - 1\big|$. Then by \eqref{netdis} and by the Claim
\eqref{claim} we get

\begin{multline*}
\int_{I_{2\de}(\Om)} | w(\hat v_{\e}) - w(\tilde
v_\e)|^2 \le C \tilde\de^2 \sum_{(i,j)\in
(I_{\tilde\de}(\Om))^1_{\tilde\de}} m_{i,j}^2 \cr \le C {\tilde
\de^2} \e \sum_{(i,j)\in (I_{\tilde\de}(\Om))^1_{\tilde\de}}
\int_{L_{i,j}} \frac{1}{2} | w'_\e|^2 + \frac{1}{2 \e^2} W ( w_\e)
\le C \tilde \de \e |\log\e|^h.
\end{multline*}

Let us pass to estimate the third addend in the right hand side of
\eqref{bo}. Let $i \in I^2_{\tilde\de}(\Om):= \{i\in \tilde\de\Z^2:
i+[0,\tilde\de]^2 \subset I_{\tilde\de}(\Om)\} $ and let $f_i: (0,\tilde
\de)^2\to \R^2$ be defined by
$$
f_i(x):=  w( \tilde v_\e(x  +  i)) -  w_\e(x + s_{\tilde
\de} +  i) .
$$
Therefore $f_i(0)=0$, and for every $x=(\bar x_1,\bar x_2)\in
(0,\tilde\de)^2$ we have
$$
|f_i(x)|^2 \le C \tilde\de \left ( \int_{0}^{\tilde\de}
|\frac{\partial}{\partial x_1} f_i(t, \bar x_2)|^2 \, dt +
\int_{0}^{\tilde \de} |\frac{\partial}{\partial x_2} f_i(0,t)|^2
\, dt \right).
$$
Integrating with respect to $x$ in  $(0,\tilde\de)^2$ we have
\begin{multline*}
\int_{(0,\tilde\de)^2} |f_i(x)|^2 \, dx  \\ \le
C \tilde\de
\left(\tilde\de \int_{(0,\tilde\de)^2} |\nabla f_i(x)|^2 \, dx +
\tilde\de^2 \int_{0}^{\tilde\de} |\frac{\partial}{\partial x_2}
f_i(0,t)|^2 \, dt \right),
\end{multline*}
and hence
\begin{multline*}
\int_{I_{2\de}(\Om)} | w( \tilde v_\e) -  w_\e(\cdot +  s_\de)|^2
\, dx \le \sum_{i\in I^2_{\tilde\de}(\Om)}
\int_{(0,\tilde\de)^2} |f_i(x)|^2 \, dx \le \\
\tilde \de^2\Big( C \int_{I_{\tilde\de}(\Om)} | \nabla w( \tilde
v_\e)|^2 + |\nabla w_\e(\cdot + s_{\tilde\de})|^2 +
\tilde \de \int_{R_{\tilde\de}} | w'( \tilde v_{\e})|^2 \\ +
|w'_\e (\cdot + s_{\tilde\de})|^2\Big) 
\le \tilde\de^2 |\log
\e|^h.
\end{multline*}
\par
We conclude by proving the Claim \eqref{claim}. Let $t \in L$ be such that
$$
m:= \max_{s\in L} ||\f|-1| = ||\f(t)|-1|,
$$
and let $(a,b)\subset L$ be the maximal interval containing $t$ such that
$ ||\f(s)|-1|\ge \frac{m}{2}$ for all $s\in (a,b)$.
Then we have
\begin{equation}\label{cl1}
\int_{L} \frac{1}{2} |\f'|^2 + \frac{1}{2 \e^2} W (\f) \ge \int_{L}
\frac{1}{2 \e^2} W (\f) \ge C\frac{ |b-a| m^2}{\e^2}
\end{equation}
if $|b-a|\ge \e/2$ we are done; otherwise either we have
$||\f(a)|-1| = \frac{m}{2}$ or $||\f(b)|-1| = \frac{m}{2}$. Then it
is very easy to see that
\begin{equation}\label{cl2}
\int_{L} \frac{1}{2} |\f'|^2 + \frac{1}{2 \e^2} W (\f) \ge \int_{L}
\frac{1}{2} |\f'|^2 \ge C \frac{ m^2}{|b-a|} \ge C \frac{ m^2}{\e},
\end{equation}
and this concludes the proof of the claim. \qed

\medskip

\subsection {Proof of $\SD_\e \preceq \XY_\e$}
As in the proof of $\XY_\e \preceq \GL_\e$, we set $\tilde\de:= \e
|\log \e|^{h+1}$, $\de=\lambda_\e\tilde\de$ with $\lambda_\e:=
\tfrac{d}{d-2\tilde\de}$, where $d$ is the distance between $0$
and $\partial \Om$. Moreover, we denote by $\hat \de=
\e[\frac{\tilde\de}{\e}]$, where $[\cdot]$ denotes the integer
part.
\par
Let$(\mu_\e)\subset X$ be  such that $\XY_{\e}(\mu_\e)\le C$. By
the very definition of $\XY_{\e}$ there exists $(v_\e)\subset
\asexy$ such that
$$
\frac{1}{|\log\e|^{h-1}} \mu_{v_\e} = \mu_\e, \qquad
\frac{1}{|\log\e|^h} XY_{\e}(v_\e) - \XY_{\e}(\mu_\e) \to 0 \quad
\text{ as } \e\to 0.
$$ 
Let $I_{\hat\de}(\Om)$ be defined as in \eqref{idelta} (with
$\delta$ replaced by $\hat\delta_\e$), and set
\begin{align}
&R^h_{\hat\de}:= \{(\e i, \hat \de j), \, i,j\in\Z\}, \\
&R^v_{\hat \de}:= \{(\hat \de i, \e j), \, i,j\in\Z\}, \\
&R_{\hat \de}:= (R^h_{\hat \de} \cup R^v_{\hat\de})\cap I_{\hat
\de}(\Om).
\end{align}
Therefore,  there exists $s_{\hat \de}\in \e\{0,1, \ldots,
\tfrac{\hat \de}{\e}\}$ such that, denoting by
$R^1_{\hat\de}(s)$ the class of nearest neighbors
$(l,k) \in  R_{\hat\de} + (s,s)$, we have
\begin{multline}\label{fuxy}
XY_{\e}(v_\e) \ge \sum_{s\in \e\{0,1, \ldots, \tfrac{\hat
\de}{\e}\}} \quad \sum_{(l,k)\in R^1_{\hat\de}(s)}
\frac{1}{2} |v_\e(l) - v_\e(k)|^2  \\
\ge \frac{\hat \de}{\e}
\sum_{(l,k)\in R^1_{\hat\de}(s_{\hat \de})} \frac{1}{2} |v_\e(l) -
v_\e(k)|^2.
\end{multline}
Let $\theta_\e(v_\e)$ be a determination of the phase of $v_\e$,
defined by the identity $v_\e(l)= e^{i  \theta_{\e}(v_\e)(l)}$ for
every $l\in\oez$. By \eqref{fuxy}, since $ XY_{\e}(v_\e)\le C|\log
\e|^h$, we immediately deduce that
$$\sup \{|v_\e(l) - v_\e(k)|, \,
(l,k)\in R^1_{\hat\de}(s_{\hat \de})\} \to 0 \text{ as } \e\to 0,
$$
so that by Taylor expansion we have

\begin{multline*}
{2\pi} \, \dist \Big( \frac{1}{2\pi} \theta_\e(v_\e)(l) -
\frac{1}{2\pi} \theta_\e(v_\e)(k) \, , \,  \Z \Big) \\=
\theta_\e(v_\e)(l) -  \theta_\e(v_\e)(k) = (1 + o(1)) |v_\e(l) -
v_\e(k)|.
\end{multline*}

\par
Let us set $A_{\tilde\de} := I_{2\tilde \de}(\Om) +
(s_{\hat\de},s_{\hat\de})$, and let $(A_{\tilde\de})^1_{\hat \de}$
be the class of $\hat\de$-nearest neighbors in $A_{\tilde\de}$. By
Jensen inequality, in view also of \eqref{fuxy}, we deduce

\begin{multline}\label{stiensd}
 \frac{4\pi^2}{2} \sum_{(l,k)\in (A_{\tilde\de})^1_{\hat \de}}
  \dist^2 \Big( \frac{1}{2\pi} \theta_\e(v_\e)(l) - \frac{1}{2\pi}
\theta_\e(v_\e)(k) \, , \,  \Z \Big) \le \\
\frac{4\pi^2}{2} \frac{\hat\de}{\e}\sum_{(l,k)\in
R^1_{\hat\de}(s_{\hat \de})}
 \dist^2 \Big( \frac{1}{2\pi} \theta_\e(v_\e)(l) -
\frac{1}{2\pi} \theta_\e(v_\e)(k) \, , \,  \Z \Big)  \\
\le (1+o(1))
XY_{\e}(v_\e).
\end{multline}

Since $\tfrac{\hat \de }{\de}\Om^0_{\de} +(s_{\hat \de},s_{\hat
\de}) \subseteq (A_{\tilde\de})^0_{\hat \de}$, we are in a
position to introduce the sequences
$$
u_\e(l):= \frac{1}{2\pi} \theta_\e(v_\e)\left(\frac{\hat \de
}{\de} l  +(s_{\hat \de},s_{\hat \de})\right) \quad \text{ for all
} l\in\Om^0_\de,$$
$$ \eta_\e:= \frac{1}{|\log \de|^{h-1}}
\mu_{u_\e},
$$
obtaining
$$
 SD_\de(u_\e) \le \frac{1}{2} \sum_{(l,k)\in
(A_{\tilde\de})^1_{\hat \de}}  \dist^2 \Big( \frac{1}{2\pi}
\theta_\e(v_\e)(l) - \frac{1}{2\pi} \theta_\e(v_\e)(k) \, , \,  \Z
\Big) .
$$

Therefore, by \eqref{stiensd} we deduce that property i) of
Definition \ref{equivalence} is satisfied. In order to check that
also property ii) of Definition \ref{equivalence} holds, it is
enough to prove that $\|\mu_\e - t_\e \eta_\e\| \to 0$ in $X$ for
some $t_\e\to 1$. Arguing as in the proof of $\XY_{\e} \preceq
\GL_{\e}$ it is enough to check that $|\log\e|^h \, \|w(e^{i\,2\pi
u_\e}) - w(v_\e)\|^2_2 \to 0$; we skip the details, that can be
easily checked by the reader. \qed

\medskip
\subsection {Proof of $\GL_\e \preceq \SD_\e$}
Since $ XY_{\e} \le 4 \pi^2 SD_{\e}$ pointwise, we immediately
deduce that $\XY_\e\preceq \SD_\e$. Therefore, the desired order
relation is obtained  by proving  $\GL_\e \preceq \XY_\e$.
\par We first  observe that by Lemma $2$ in \cite{AC},
there exists a constant $C>0$ such that, for every $v\in \asexy$
\begin{equation*}
\frac{C}{\e^2} \int_{\Om_\e} W(w(v)) \le  XY_\e(v).
\end{equation*}
Let $\lambda_\e \nearrow 1$ be such that $\lambda_\e \Om\subset
\Om_\e$ (we recall that $\Om$ is star-shaped), and such that $\lambda_\e\ge 1-\frac{\e}{c}$ for some constant $c$. Given a sequence
$t_\e\searrow 0$ we have
\begin{equation}\label{stimaac0}
(1+t_\e) XY_{\e}(v) \ge \int_{\lambda_\e \Om} \Big(\frac{1}{2}
|\nabla w(v)|^2 + \frac{C t_\e}{\e^2} W(w(v))\Big) \, dx,
 \end{equation}
for every  $v\in \asexy.$
Let now $\mu_\e$ be a sequence such that $\XY_\e(\mu_\e) \le C$,
and let $v_\e$ be such that $\tfrac{1}{|\log\e|^h} XY_\e(v_\e) -
\XY_\e(\mu_\e)\to 0$. Set then $\de = \e$,
$$
w_\e(x):= w(v_\e (\lambda_\e x))  \text{ for every } x\in \Om,
\qquad \eta_\e = \frac{1}{|\log\e|^{h-1}} \frac{J w_\e}{\pi}.
$$
Then by \eqref{stimaac0} we get
\begin{equation}\label{stimaac}
XY_{\e}(v_\e) \ge \frac{1}{1+ t_\e}GL^{Ct_\e}_{\e} (w_\e).
\end{equation}

From \eqref{stimaac}, we easily deduce that $\GL_\e^{Ct_\e}
\preceq \XY_\e$. Indeed,  Property i) of Definition
\ref{equivalence} is a direct consequence of \eqref{stimaac},
while the  proof of Property ii) follows as in the proof of
$\XY_\e \preceq \GL_\e$. Finally, choosing $t_\e\to 0$ such that
$s_\e:= C t_\e$ satisfies the assumptions of Proposition \ref{se},
we conclude that  $\GL_\e \simeq \GL_\e^{Ct_\e} \preceq
\XY_\e\preceq \SD_\e$. \qed

\begin{oss}\label{newproof}
In \cite{JS} it is proved that for $h=1$ the functionals $\GL_\e$
in \eqref{glene} are equi-coercive, and $\Gamma$-converge to the
functional $\pi |\mu|(\Om)$. In \cite{Po}[Theorem 3.4]) the same
$\Gamma$-convergence result is proved for the functionals $\SD_\e$
defined in \eqref{sdene}. In view of this  result, of Theorem
\ref{mainthm} and of Theorem \ref{consequi}, we obtain a new proof
for the compactness of the jacobians given in \cite{JS}, and of
the corresponding $\Gamma$-convergence result of Ginzburg-Landau
functionals in the logarithmic regime.
\end{oss}

\begin{oss}\label{ball}
For latter use we observe that, if we restrict the Ginzburg-Landau
functionals $GL_\e$ to the  fields $w_\e$
 valued in $B_1$, then the equivalence result stated
in Theorem \ref{mainthm} still holds true. Actually, given a
sequence $w_\e$ with finite energy we can always project it on
$B_1$, without increasing its energy, and without changing the
limiting behaviour of the corresponding topological singularities.
\par
 This $L^\infty$-bound will simplify the
proof of compactness properties of the quantity  $j_\e$ associated
with  $w_\e$, we will deal with in the next Section.
\end{oss}

\section{New results for the asymptotic of  $ SD_{\e} $}\label{New SD}

As explained in the Introduction (see also Remark \ref{newproof}),
the $\Gamma$-limit of the functionals $\SD_\e$ is known only for
$h=1$. On the other hand, the analogous result for the
Ginzburg-Landau functionals $\GL_\e$ has been proved by Jerrard
and Soner in \cite{JS2} for all values of $h$ ($h=1$ and $ h = 2$ being the
most relevant cases). In this section we use the variational
equivalence argument to deduce $\Gamma$-convergence results for
the screw dislocation model in the $|\log\e|^2$ scaling regime. We recall that this energy
scaling has been already considered  in the context of interacting
edge dislocations in \cite{GLP},  providing in the limit a
macroscopic strain gradient model for plasticity.
 We will extend
this result to our discrete model of screw dislocations without
imposing, as in \cite{GLP}, that the minimal distance between the dislocations is of order $\rho>>\e$ .
\par
Before giving the rigorous results, let us explain by  heuristic arguments why the $|\log\e|^2$ energetic regime is somehow critical, and hence gives rise to an interesting macroscopic limit. 
Assume that in the crystal there is a distribution $\mu$ of  a certain number $N_\e$ of screw dislocations of unit length. The self energy of the system is, in first approximation,  proportional to  $N_\e|\log\e|$. On the other hand, each dislocation induces also  a far field:  the macroscopic strain field $\beta$ has to satisfy the kinematic constrain $\curl\beta =  \mu$. The elastic energy depends quadratically on $\beta$, and then it is proportional to $N_\e^2$. We deduce that if $N_\e\approx |\log\e|$ then the self energy and the far field energy corresponding to the macroscopic strain $\beta$ (namely the interaction energy)  are of the same order $|\log \e|^2$. Therefore, in the limit as $\e\to 0$, the energy is given by the sum of these two contributions, a self energy, one-homogeneous with respect to the dislocation density, and an interaction energy, quadratic with respect to the macroscopic fields $\beta$'s satisfying the kinematic relation $\curl\beta = \mu$. This kind of energy can be settled in the recent strain gradient theories for plasticity introduced in \cite{FH}.
\par
We underline that, as it will be clear in our analysis (see Theorem \ref{mthmj}) the same result holds true
for the $XY$ spin system model.
\subsection{The $\Gamma$-convergence result for $\GL_\e$ in the $|\log\e|^2$ regime }
Here we recall the $\Gamma$-convergence result for the functionals
$\GL_\e$ in the energetic regime  $GL_\e \approx |\log\e|^2$
corresponding to $h=2$ given by Jerrard and Soner in \cite{JS2}. For
the sake of simplicity we will specialize the results assuming that
the order parameters $w_\e$ take values in $B_1$ (see Remark \ref{ball}).
\par Given
$w\in H^1(\Om;B_1)$, set $$ j(w):= w\times \nabla w = \Big ( w_1
(w_2)_{x_1} - w_2 (w_1)_{x_1} , w_1 (w_2)_{x_2} - w_2 (w_1)_{x_2}
\Big).
$$
Note that by definition we have $ J(w)= \tfrac{1}{2} \curl j(w)$. Consider the functionals $ \bGL_\e:
X \times L^2(\Om;\R^2) \to [0,+\infty]$ defined as follows
\begin{multline}\label{tgl}
 \bGL_\e(\mu,j):= \\
 \inf\left\{ \frac{1}{|\log\e|^2} GL_\e(w), \,
w\in H^1(\Om;B_1):  \frac{j(w)}{|\log\e|}=j, \,
\frac{J(w)}{\pi |\log\e|}=
 \mu \right\}.
\end{multline}
By \cite[Theorem 1.1 and 1.2]{JS2} we deduce the
following $\Gamma$-convergence result

\begin{theorem}[Jerrard and Soner, 2002]\label{mthmjs2}
The functionals $\bGL_\e:X\times L^2(\Om;\R^2)\to [0,+\infty]$ defined in
\eqref{tgl} are equi-coercive: if $(\mu_\e,j_\e)$ is a
sequence with bounded energy then, up to subsequences, $\mu_\e \to
\mu$ for some $\mu \in X$ and $j_\e \weak j$  weakly in
$L^2(\Om;\R^2)$, for some $j\in L^2(\Om;\R^2)$.
\par Moreover,  $\bGL_\e$  $\Gamma$-converge (with respect to the same topology) to the functional
$ \bGL:X\times L^2(\Om;\R^2) \to [0,+\infty]$, defined as
\begin{equation}\label{glglejs2}
\bGL(\mu,j):= \pi |\mu|(\Om) + \frac{1}{2}  \int_{\Om}|j|^2,
\end{equation}
if  $\mu$ is a measure in  $H^{-1}(\Om)$ and 
$\curl j= 2 \pi \mu$,
and infinity elsewhere.
\end{theorem}
From Theorem \ref{mthmjs2} we immediately deduce the following
\begin{theorem}\label{mthmjs3}
The functionals $\GL_\e:X\to [0,+\infty]$ defined in \eqref{glene} with
$h=2$ are equi-coercive and $\Gamma$-converge, as $\e\to 0$, to
the functional $\GL:X\to [0,+\infty]$ defined by
\begin{equation*}
\GL(\mu):= \pi |\mu|(\Om) + \frac{1}{2} \inf \Big\{
\int_{\Om}| j |^2 \, dx, \, j \in L^2(\Om;\R^2): \,  \curl j =
2\pi \mu\Big\},  
\end{equation*}
if $\mu$ is a measure in  $H^{-1}(\Om),$
and infinity elsewhere.
\end{theorem}

\subsection{New results for homogenizing dislocations in the $|\log\e|^2$ regime}
From the variational equivalence between $\GL_\e$ and $\SD_\e$
stated in Theorem \ref{mainthm} (see also Remark \ref{ball}), we
deduce the equivalent result stated in Theorem \ref{mthmjs3} for
the energy functionals corresponding to screw dislocations.
\par
For the reader convenience, we state the $\Gamma$-convergence
result for the dislocation energy functionals $\F_\e$ defined in \eqref{def:F_e}
according with \eqref{sdene} with $h=2$, but without the
pre-factor $4\pi^2$ that has been useful to compare the $SD$ model with
$XY$ and $GL$ models, but which has not physical meaning.

\begin{theorem}\label{mthmsd}
The functionals $\F_\e:X\to [0,+\infty]$, defined by
\begin{equation}\label{def:F_e}
\F_\e(\mu):=\frac{1}{ |\log \e|^2} \inf \left\{ SD_\e(u), \, u\in
\asesd: \frac{\mu_u}{|\log\e|}  =\mu  \right\},
\end{equation}
are equi-coercive and $\Gamma$-converge as $\e\to 0$ to the
functional $\F:X\to [0,+\infty]$ defined by
$$
\F(\mu):= \frac{1}{4\pi} |\mu|(\Om) + \frac{1}{2} \inf \Big\{
\int_{\Om}|\beta|^2, \beta\in L^2(\Om;\R^2), \,  \curl \beta =
\mu\Big\},
$$
 if $\mu$ is a measure in  $H^{-1}(\Om)$, and infinity elsewhere.
\end{theorem}
\begin{proof}
The proof is a straightforward consequence of Theorem \ref{mainthm} and Theorem \ref{mthmjs3}.
\end{proof}

In order to give the analogue of Theorem \ref{mthmjs2} for the
$XY$ and the screw dislocations model, let us associate to any
discrete strain $\beta^e_{u}$, with $u\in \asesd$, its
corresponding  piecewise constant strain field  $\hat
\beta^e_{u}:= \pr (\nabla w(u))$ in $\Om_\e$ (where $\pr$ is
defined component by component as in \eqref{pr}), and extend it to
zero in $\Om\setminus \Om_\e$. Moreover, given $v\in \asexy$,
we set $\hat j_{v}:=2\pi \pr (\nabla
w(\tfrac{1}{2\pi}\theta(v)))$ in $\Om_\e$, so that $\hat
j_{v} = 2\pi \hat\beta^e_{\theta(v)}$,
 and extend it to zero in $\Om\setminus
\Om_\e$.
\par
We are in a position to introduce  the functionals $ \bXY_\e :
X\times L^2(\Om;\R^2) \to [0,+\infty]$ defined as
\begin{multline}\label{txw}
 \bXY_\e (\mu,j):= \\
 \inf_v
 \left\{
 \frac{1}{|\log\e|^2} XY_\e(v), \,
v\in \asexy(\Om):   \frac{\hat j_{v}}{|\log\e|}=j, \,
\frac{\mu_{v}}{|\log\e|}=\mu
 \right\}.
\end{multline}

and the functionals $ \bSD_\e : X\times L^2(\Om;\R^2)\to [0,+\infty]$ defined as

\begin{multline}\label{tsd}
\bSD_\e(\mu, j):=\\  \inf_u \left\{\frac{4\pi^2}{|\log\e|^2}
SD_\e(u), \, u\in \asesd(\Om): \frac{2\pi \hat
\beta^e_{u}}{|\log\e|} = j, \, \frac{\mu_{u}}{|\log\e|} = \mu
\right\}.
\end{multline}
The following Theorem establishes the variational equivalence for
the functionals $ \bGL_\e,\, \bXY_\e$ and $\bSD_\e$ with respect to the strong topology in $X$
and the weak topology in $L^2(\Om;\R^2)\to \R$.

\begin{theorem}\label{mthmj}
The functionals $ \bGL_\e,\, \bXY_\e$ and $\bSD_\e$ are variationally equivalent.
\end{theorem}

\begin{proof} The proof  follows the lines of Theorem \ref{mainthm}.
One has only to check that for each change of variables involved
in the proof, the corresponding fields $j(w_\e)$,
$\hat j_{v_\e}$ and $2\pi\hat \beta^e_{u_\e}$, rescaled by
$|\log\e|$ share the same weak limit in $L^2$. Let us check it
only for the order relation $ \bXY_\e \preceq \bGL_\e$, the other
order relations being analogous.
\par
To this purpose, let $(\mu_\e,j_\e)$ be such that $\bGL_\e
(\mu_\e,j_\e)\le C$, and let $w_\e$ be such that
$$
\frac{J(w_\e)}{\pi |\log\e|} =\mu_\e, \quad
\frac{j(w_\e)}{|\log\e|}
 = j_\e, \quad  \frac{1}{|\log\e|^2} GL(w_\e) -
\bGL_{\e}(\mu_\e,j_\e)  \to 0  \text{ as } \e\to 0.
$$

Let now $\de$  and $v_\e$ be as in the proof of the order
relation $\XY_\e\preceq \GL_\e$ in Theorem \ref{mainthm} (with
$h=2$), satisfying
$$
XY_\de(v_\e) \le GL_\e(w_\e) + o(1), \quad \|w(v_\e) -
w_\e\|_2\to 0, \quad \frac{1}{|\log\de|} \mu_{v_\e} - \mu_\e\to
0.
$$
In order to conclude, it is left to prove that
\begin{equation}\label{concl}
\frac{1}{|\log\e|} (j (w_\e) - \hat j_{v_\e})\weak 0 \text{
in } L^2(\Om;\R^2).
\end{equation}
 We first prove that
\begin{equation}\label{convj}
\frac{1}{|\log\e|} (j(w_\e) - j(w(v_\e)) \weak 0 \quad \text{ in
} L^2(\Om, \R^2).
\end{equation}
Since we have
$$ \frac{1}{|\log\e|} (j(w_\e) - j(w(v_\e))) =
(w_\e - w(v_\e)) \times \frac{\nabla w_\e}{|\log\e|} +
\frac{w(v_\e)}{|\log \e|} \times \nabla (w_\e - w(v_\e)),
$$
by Holder inequality and by integration by parts it is easy to
deduce \eqref{convj}.
\par
Now, to obtain \eqref{concl} it remains to check that
$\tfrac{1}{|\log\de|} (j (w(v_\e)) - \hat j_{v_\e}) \weak
 0$ in $L^2(\Om;\R^2)$. To this purpose, set $\rho_\e:= \de|\log \de|^2$, and let
 $\Om_{\de,\rho_\e}$ be the union of $\de$-squares $Q_i$ in $\Om_\de$ such
 that the oscillation of $w(v_\e)$ on $Q_i$ is bounded by $\rho_\e$.
 Since $XY_\de (v_\e) \le C|\log\de|^2$, it easily follows that
 $$
 |\Om \setminus \Om_{\de,\rho_\e}|\le \frac{C}{|\log\de|^2}\to 0 \quad\text{ as } \e\to 0.
 $$
Therefore, since
 $\tfrac{1}{|\log\de|}j(w(v_\e))$ and  $\tfrac{1}{|\log\de|} \hat j_{v_\e}$ are bounded  in
 $L^2(\Om;\R^2)$,  we get
 \begin{equation}\label{zeromeasure}
\frac{1}{|\log\de|}\Big( j(w(v_\e)) - \hat j_{v_\e} \Big)
\Big(1-\chi_{\Om_{\de,\rho_\e}}\Big) \weak 0  \text{ in }
L^2(\Om;\R^2) \text{ as } \e\to 0.
\end{equation}
To conclude the proof of \eqref{concl} it remains to show that
\begin{equation}\label{stifi}
 \lim_{\e\to 0}
\|j(w(v_\e) -
\hat j_{v_\e}\|_{L^\infty(\Om_{\de,\rho_\e};\R^2)} = 0.
\end{equation}
To this purpose, notice that,  on each $\de$-square $Q_i\in
\Om_{\de,\rho_\e}$ we have  (for $\e$ small enough)
$\hat j_{v_\e}= \nabla w(\theta(v_\e))$. Therefore by
Taylor expansion
$$
|j(w(v_\e)) - \hat j_{v_\e}| = |j(w(v_\e)) - \nabla
w(\theta(v_\e))|\le C \frac{\rho_\e^2}{\de}= C \de |\log\de|^4\to
0  
$$
 on $ \Om_{\de,\rho_\e}$, that clearly implies  \eqref{stifi}, and this concludes the proof
of \eqref{concl}.
\par
Essentially, the same arguments can be used to prove all the order
relations between the functionals $ \bGL_\e,\, \bXY_\e$ and
$\bSD_\e$, so that we prefer to skip the details.
\end{proof}

From Theorem \ref{mthmj} we immediately deduce the analogous
$\Gamma$-convergence result for the screw dislocation functionals.
In particular,  we obtain a new result in the context of
homogenizing dislocations, generalizing the results in \cite{GLP}
to the scalar case of discrete interacting screw dislocations.
 For the reader convenience, we state the result introducing the
dislocation energy functionals $\G_\e$, defined in \eqref{fine} according with
\eqref{tsd}, but without the pre-factor $2\pi^2$ in the
strain, and without the prefactor $4\pi^2$ .

\begin{theorem}\label{mthmsd2}
The functionals $\G_\e: X\times L^2(\Om;\R^2)\to [0,+\infty]$ defined by
\begin{multline}\label{fine}
\G_\e(\mu,\beta):= \\
 \inf_u \left\{\frac{1}{|\log\e|^2} SD_\e(u),
\, u\in \asesd(\Om): \frac{ \hat \beta^e_{u}}{|\log\e|} =\beta, \,
\frac{\mu_{u}}{|\log\e|} = \mu \right\},
\end{multline}
 are equi-coercive and $\Gamma$-converge, with respect
to the strong convergence in $X$ and the weak convergence in
$L^2(\Om;\R^2)$, to the functional $ \G:X\times L^2(\Om;\R^2) \to
[0,+\infty]$ defined by
\begin{equation}\label{maene}
 \G(\mu):= \frac{1}{4\pi} |\mu|(\Om) + \frac{1}{2}
\int_{\Om}|\beta|^2, 
\end{equation}
if  $\mu$ is a measure in 
$H^{-1}(\Om) \text{ and } \curl \beta =  \mu$,
and $+\infty$ elsewhere.
\end{theorem}
The $\Gamma$-limit $\G$ in \eqref{maene} represents the
macroscopic energy corresponding to a density of dislocations
$\mu$ and a macroscopic strain $\beta$. The first term in $\G$
represent the self energy of the dislocation density, it is
$1$-homogeneous, and therefore it is compatible with concentration
of dislocations. Notice that the compatibility
condition $\mu=\curl\beta$ agrees with concentration of the
density of dislocations, whenever $\mu\in H^{-1}(\Om)$. This is
the case of concentration on lines, according with the well known
configuration of wall dislocations.
The second term is the elastic energy of the
macroscopic strain $\beta$, and represents an interaction energy
of the dislocations. This kind of macroscopic
energies, depending  also on the derivatives of the strain
$\beta$, are usually referred to as  strain gradient
theories in plasticity and they have been introduced in \cite{FH}. The energy $\G$ derived by $\Gamma$-convergence represents then a justification of such phenomenological theory, and provides explicit   self and the interaction energy densities.

\section{Further extensions and conclusions}\label{comments}
In this paper we have investigated the variational equivalence of
some model characterized by the presence of  topological
singularities. It is natural to ask if our method can be extended
to other contexts in the huge field of modelling singularities and
in particular dislocations.
\par
In this Section we propose some extension of our approach. We first
consider the so called {\it core radius approach} to
dislocations in the two dimensional setting and then we 
investigate the case of three dimensional dislocations.

\subsection{The core radius approach}
In order to deal with the singularity of the stress field around a
dislocation, a very fruitful approach consists in removing a
region of size $\e$ around each dislocation, usually referred to
as {\it core region}, obtaining in this way a $L^2$ integrable
field on the complementary domain. This method has been exploited
very recently in variational models for dislocations \cite{CL},
\cite{Po}, \cite{GLP}. Its feature is that the discreteness
of the problem is carried by the length-scale $\e$, representing
the atomic distance, while the mathematical framework is
continuous.
\par
Within the core radius approach, many mathematical details have to
be fixed in order to make the $\Gamma$-convergence problem well
posed and doable: for instance in \cite{Po} it has been introduced
a small penalization for the number of $\e$-disks removed by the
domain; such a penalization plays the role of the potential term
$W$ in the Ginzburg-Landau energy, ensuring that the number of
dislocations is bounded by $|\log\e|$, and in particular that the
measure of the core region tends to zero as $\e\to 0$.
\par
 This approach, as proposed by Bethuel, Brezis and H\'elein,  can be
very fruitful also as a variant of the Ginzburg-Landau approach to
vortices. The variational equivalence argument introduced in this
paper seems to be a natural tool to compare the core radius
approach to dislocations and vortices with pure discrete
approaches, that we believe to be equivalent.
\par
Finally, we aim to comment that the core radius approach has been
proposed in \cite{BBH} in order to compute the {\it renormalized
energy}, {\it i.e.},  the lower order term in the energy of minimizers
of the Ginzburg-Landau functionals in the logarithmic regime. This
inspired the work in \cite{CL}, where the authors compute the
Taylor expansion of  the elastic energy of edge dislocations in a
plane. The first term in the Taylor expansion is the self energy,
that for screw dislocations is given (up to a pre-factor) by
$|\mu|(\Om)$ (see Remark \eqref{newproof}). The lower order term,
corresponding to the renormalized energy, depends on the mutual
distance of the dislocations, and models the Peach-K\"ohler
attractive and repulsive forces between dislocations. In our
opinion it would be very interesting to investigate the
equivalence of this renormalized energy for vortices and
dislocations within a $\Gamma$-convergence analysis. In order to
do that, it seems necessary to exploit the equivalence between
vortices and dislocations for the first order term in the
$\Gamma$-limit expansion of the energy functionals. Indeed,
Definition \ref{equivalence} could be generalized to compare lower
order terms in the energy, in the spirit of the theory introduced
by Braides and Truvskinovsky in \cite{BrTr}.  The analysis required
to compare the renormalized energy for vortices and dislocations
seems to be very challenging.

\subsection{Three dimensional dislocations}
Screw dislocations are essentially straight dislocations lines
with parallel Burgers vector. The  general case accounts more
complexity, dealing with general  closed loops of dislocations
with given Burgers vectors.
\par
Here we aim to introduce the elastic energy in this three
dimensional setting, and the corresponding Ginzburg-Landau energy
functionals, formally obtained arguing in analogy with what  we
have done in the anti-planar setting. We will not address the
problem of proving rigorous  equivalence results between elastic
energies and Ginzburg-Landau energies in this three dimensional
context.
\par
 We consider the basic case of a
cubic lattice, whose elasticity tensor we denote by $\C\in
\R^9\times\R^9$. Let $\Om\subset \R^3$ be the reference
configuration of the three dimensional crystal.
 The displacement is now a vectorial function from $\oez:=\Om\cap
\e\Z^3$ to $\R^3$, and,  following the approach in \cite{ArOr},
 the dislocations are defined on the class of two-cells
of the crystal. In this three dimensional case, we have three kind
of two-cells, corresponding to the three different slip planes of
the cubic lattice. A basic dislocation loop is identified with a
pair,  given by a two cell and a Burgers vector $b$, that in the
cubic case belongs to the canonical base of $\R^3$. We have then
three kind of topological singularities, corresponding to the
three different slip planes, and each kind of singularity has a
vector nature, being associated to a Burgers vector in $\R^3$. It
is then natural to set up a Ginzburg-Landau model for dislocations
considering the vector valued maps taking values near a three
dimensional torus.
\par
More precisely, consider the three dimensional torus $\Torus:=
\R^3/(\e\Z^3)$, {\it i.e.}, $\R^3$ where we identify $x$ and $y$ if and
only if $x-y\in \e\Z^3$, and consider the fields $w:\Om \to
\Torus\times [0,+\infty)$. We denote by $u\in \Torus$ the first three components of
$w$, playing the role of angular components of a continuous field in classical Ginzburg-Landau theories, and representing in our model the displacement function, and we denote by $\rho$ the fourth one, representing the radial component of $w$. Therefore,   for every $w\in H^1(\Omega,\Torus\times [0,+\infty))$
we define the Ginzburg-Landau energy functionals by
\begin{equation}\label{gltorus}
GL_\e(w):=\int_\Omega \rho^2 <\C \nabla u:\nabla u> + |\nabla \rho|^2 +
\frac{1}{\e^2} (1-\rho)^2 \, dx.  
\end{equation}
In our opinion such functionals provide a good material dependent
Ginzburg-Landau model for dislocations, and this could be
justified by showing that they are indeed equivalent to suitable
discrete elastic energy functionals, defined for instance
according to the formalism introduced in \cite{ArOr}. The rigorous
formalization and proof of such a statement  would require a
specific analysis, that is not addressed in this paper. 
Clearly, more general crystal lattices could be considered within this approach.
\par
Exploiting this Ginzburg-Landau approach to dislocations provides
a  motivation to investigate the asymptotic behaviour of these
Ginzburg-Landau functionals as $\e\to 0$, generalizing in this way
the analysis done in \cite{ABO},\cite{BBH},\cite{JS}, \cite{JS2}, \cite{S}, \cite{San_ser} to the
case of vector valued fields whose singularities belong to a given
group, and whose energy is not isotropic, and not even coercive.

\section*{Acknowledgments}
The authors wish to thank Giovanni Alberti  and Stefan M\"uller for
interesting and fruitful discussions on the subject of the paper. \par The work by M. Cicalese 
was partially supported by the European Research Council under FP7, 
Advanced Grant n. 226234 ''Analytic Techniques for Geometric and Functional Inequalities''

\address{
Roberto Alicandro\\
DAEIMI, Universit\`a di Cassino\\  via Di Biasio 43, 03043 Cassino (FR), Italy\\
e-mail: {\it alicandr@unicas.it}
\and
Marco Cicalese\\
Dipartimento di Matematica e Applicazioni ``R.~Caccioppoli"\\
Universit\`a di Napoli ``Federico II"\\
 via Cintia, 80126 Napoli, Italy\\
e-mail: {\it cicalese@unina.it}
\and
Marcello Ponsiglione\\
Dipartimento di Matematica ``G. Castelnuovo" \\
Universit\'a di Roma ``La Sapienza"\\
Piazzale A. Moro 2, 00185 Roma, Italy\\
 e-mail: {\it ponsigli@mat.uniroma1.it}
}

\end{document}